\documentclass[12pt, a4paper]{article}

\makeatletter\def\input@path{{local-chicago/}}\makeatother

\usepackage[utf8]{inputenc}        % UTF-8 encoding
\usepackage[T1]{fontenc}           % Font encoding
\usepackage{times}                 % Times New Roman font
\usepackage{geometry}              % Page layout

\PassOptionsToPackage{update=none}{epstopdf}
\PassOptionsToPackage{off}{auto-pst-pdf}
\usepackage{graphicx}
\DeclareGraphicsExtensions{.pdf,.png,.jpg,.jpeg}

\usepackage{amsmath, amssymb}      % Mathematical symbols
\usepackage{booktabs}              % Enhanced tables
\usepackage{hyperref}              % Hyperlinks
\usepackage{fancyhdr}              % Custom headers and footers
\usepackage{abstract}              % Custom abstract
\usepackage{setspace}              % Line spacing
\usepackage{amsthm}
\usepackage{authblk}               %Author and affiliation

\usepackage{subcaption}
\usepackage{booktabs}
\usepackage{adjustbox}

\usepackage{subcaption}            % sub figure
\usepackage{float}                 % fix float figure

\usepackage{booktabs}
\usepackage{siunitx}               % for S‑column alignment

\usepackage[ruled,vlined]{algorithm2e}

\usepackage{color} 
\usepackage[dvipsnames]{xcolor}
\usepackage{ulem}                  % Package for strikethrough text

\usepackage[
  backend=biber,       % use biber backend
  authordate,          % Chicago Author-Date style
  maxcitenames=2,      % truncate in-text citations after 2 names
  mincitenames=1,      % keep at least 1 name before "et al."
  uniquelist=false,    % don't force extra names to disambiguate
  sorting=nyt,          % sort bibliography: name, year, title
  dashed=false
]{biblatex-chicago}

\usepackage{xurl} % better URL line breaks
\appto{\bibsetup}{\raggedright\setlength{\bibhang}{2em}} % tidy bibliography
\theoremstyle{plain}
\newtheorem{theorem}{Theorem}[section]

\newtheorem{proposition}[theorem]{Proposition}

\newtheorem{assumption}{Assumption}[section]
\newtheorem{lemma}{Lemma}

\theoremstyle{remark}
\newtheorem{remark}{Remark}[section]

\usepackage[font=footnotesize,labelfont=bf]{caption}
\hypersetup{
    colorlinks=true,
    linkcolor=blue,
    filecolor=magenta,      
    urlcolor=cyan,
}

\title{Denoised IPW-Lasso for Heterogeneous Treatment Effect Estimation in Randomized Experiments}
\author[1]{Mingqian Guan\thanks{mingqianguan@gmail.com}}
\author[2]{Komei Fujita}
\author[1]{Naoya Sueishi\thanks{sueishi@econ.kobe-u.ac.jp}}
\author[2]{Shota Yasui}

\affil[1]{Kobe University}
\affil[2]{CyberAgent}

\date{}

\begin{document}

\maketitle

\begin{abstract}
This paper proposes a new method for estimating conditional average treatment effects (CATE) in randomized experiments. We adopt inverse probability weighting (IPW) for identification; however, IPW-transformed outcomes are known to be noisy, even when true
propensity scores are used. To address this issue, we introduce a noise reduction procedure and estimate a linear CATE model using Lasso, achieving both accuracy and interpretability. We theoretically show that denoising reduces the prediction error of the Lasso. The method is particularly effective when treatment effects are small relative to the variability of outcomes, which is often the case in empirical applications. Applications to the Get-Out-the-Vote dataset and Criteo Uplift Modeling dataset demonstrate that our method outperforms fully nonparametric machine learning methods in identifying individuals with
higher treatment effects. Moreover, our method uncovers informative heterogeneity patterns that are consistent with previous empirical findings.

    \textbf{Keywords}: Causal Inference; Conditional Average Treatment Effect; Uplift Modeling
\end{abstract}

\section{Introduction}
\label{sec:introduction}

Understanding heterogeneous treatment effects is a central concern in causal inference, as it provides critical insights into effective policy design and the underlying mechanisms driving treatment responses. Although average treatment effects (ATE) provide an overall summary of treatment efficacy, they fail to capture important variation between individuals. In many real-world applications, such as policy targeting or personalized marketing, identifying those who benefit most (or are harmed) by a treatment can improve resource allocation and help to avoid unintended consequences.

Recent advances in the analysis of heterogeneous treatment effects have occurred in the development of machine learning-based methods for estimating conditional average treatment effects (CATE). A prominent approach is the meta-learning framework (e.g., \cite{kunzel2019metalearners}), which leverages a wide range of supervised learning methods for CATE estimation, including the X-learner \parencite{kunzel2019metalearners}, R-learner \parencite{nie2021quasi}, and DR-learner \parencite{kennedy2023towards}.
Another line of research tailors machine learning methods for causal inference tasks.
These include tree-based methods (e.g., the causal forest of \cite{wager2018causalforest}) and deep learning methods (e.g., \cite{shalit2017estimating}; \cite{shi2019adapting}).
These advances have primarily focused on observational settings, with less emphasis on randomized experiments.

The present study proposes a novel estimation method for CATE using randomized controlled trial (RCT) data to address two core concerns in CATE estimation: accuracy and interpretability. Accurate estimation of CATE enables precise targeting of treatments, which reduces the risk of misguided decisions. Interpretability is equally important, as it ensures transparency in decision-making and enables policymakers to understand the mechanisms underlying treatment effect heterogeneity.
However, contemporary machine learning approaches often produce black-box models with limited interpretability, posing risks in decision-making contexts.

The proposed estimation method, referred to as the denoised IPW (DIPW)-Lasso, builds on the inverse probability weighting (IPW) transformation. In experimental settings where propensity scores are known, the IPW transformation simplifies CATE estimation to a standard regression problem. Assuming that CATE is a linear function of covariates, we estimate it using the Lasso \parencite{tibshirani1996regression}. A major drawback of the IPW transformation, however, is that it introduces substantial noise, even when the true propensity scores are used. To address this issue, we propose a noise reduction procedure that preserves the advantages of IPW while significantly improving estimation accuracy. The procedure simply adds terms orthogonal to the covariates to the regression. We show that the augmented IPW (AIPW) transformation arises as a special case of our noise reduction method.

The inefficiency of IPW-type estimators has been recognized in the literature, particularly in the context of ATE estimation (e.g., \cite{hirano2003efficient}). However, we observe that this inefficiency becomes more pronounced when treatment effects are small relative to the variability of outcomes.
This situation is common in practice. For instance, \textcite{fryer2017production} conducted a meta-analysis of 105 school-based RCTs and found that the average effects of interventions on students' math and reading achievement were approximately \(0.05\sigma\) and \(0.07\sigma\), respectively.
Similar evidence has been documented in other domains, including medicine (e.g., \cite{susukida2025application}) and marketing (e.g., \cite{korkames2025meta}).

The DIPW-Lasso is particularly effective in settings where treatment effects are modest, providing both high estimation accuracy and interpretability. This advantage stems from two key components: the noise reduction procedure and the structural assumption on CATE. The former improves the stability of IPW-based estimation, while the latter contributes to both precision and interpretability. Recent empirical findings support the utility of such structural assumptions in noisy settings. In their analysis of the effects of nudges, \textcite{athey2025machine} demonstrated that a model incorporating structural assumptions outperformed a fully nonparametric one in terms of targeting. 

We theoretically show that the DIPW-Lasso improves upon its non-denoised counterpart in terms of prediction error. 
An important aspect of the result is that, although our noise reduction method requires estimating a nuisance function, such as the outcome regression function, it relies on only weak assumptions regarding the properties of the estimators.
Simulation results indicate that denoising has a substantial effect when treatment effects are small.
In our setting, the noise reduction method reduced the root mean squared error (RMSE) to approximately one-third to one-quarter of that without noise reduction.
Furthermore, our method demonstrated superior performance compared to recently proposed CATE estimation methods: the T-learner, X-learner, R-learner, DR-learner, and causal forest.

Our empirical results provide further support for the effectiveness of the proposed method.
We analyzed two real-world datasets: the Get-Out-the-Vote (GOTV) dataset \parencite{gerber2008social} and the Criteo Uplift Modeling dataset  \parencite{diemert2018large}.
Since true treatment effects are unobservable in real data, we evaluated the performance of estimation methods using the uplift curve, a commonly used tool in fields such as marketing and personalized medicine for assessing whether individuals with high treatment effects are properly identified (see, e.g., \cite{gutierrez2017causal}).
We found that the methods that impose linearity on CATE exhibited better targeting performance than nonparametric methods, such as the causal forest, that make no functional form assumptions about the CATE function.
Moreover, we empirically confirmed that our method substantially reduced the noise even when the outcomes are binary.
These results demonstrate the combined benefit of noise reduction and imposing a linear structure on CATE.

A limitation of our method is that it requires knowledge of the propensity score. Therefore, it cannot be applied to observational data. However, although RCTs are considered the gold standard for causal inference, there is relatively little research on CATE estimation methods using experimental data.
Consequently, there is currently no established guideline or consensus on which method should be used for estimating the CATE with experimental data.
Our proposed method provides a valuable option for estimating CATE using experimental data.

\subsection{Related Literature}

While research on CATE estimation using RCT data is limited, a few studies are particularly relevant to ours. \textcite{imai2013estimating} framed CATE estimation as a variable selection problem under the assumption that CATE is a linear and sparse function of covariates.
A key limitation of their approach is that it requires a correct specification of the baseline outcome, whereas our method does not.

Our approach is also closely related to that of \textcite{tian2014simple}.
They imposed linearity and sparsity assumptions on CATE and proposed a method to improve the efficiency of pseudo-outcome regression.
A central distinction lies in the strategy for variance reduction: our method performs denoising directly on the IPW-transformed outcome, whereas theirs focuses on reducing the variance of the OLS estimator of the linear CATE model by applying denoising to the corresponding score function.
Consequently, their efficiency augmentation procedure depends not only on the specification of CATE, but also on the particular estimation method used.
Theoretical validity under alternative estimators such as the Lasso has not been established.

The present study addresses a relatively underexplored aspect of CATE estimation: the application of variance reduction techniques to improve estimation precision.
Variance reduction is a well-established strategy in the context of ATE estimation under RCTs (e.g., \cite{freedman2008regression, lin2013agnostic}) and is widely recognized to significantly improve efficiency.
Recent developments have integrated regression adjustment with machine learning methods to further advance variance reduction (e.g., \cite{wager2016high}; \cite{mlvariancereductionOnlineExperiment}).
Our findings highlight that variance reduction significantly improves the precision of CATE estimation, particularly when treatment effects are small, extending its known benefits beyond ATE estimation.

A key motivation of this study is to develop a method for supporting decision-making using interpretable models. As machine learning is increasingly applied in high-stakes areas such as policy and medicine, post hoc explanation methods such as LIME \parencite{ribeiro2016should} and SHAP \parencite{lundberg2017unified} have been developed to provide interpretability for black-box models.
However, recent studies have raised concerns about the limitations and potential risks of such techniques \parencite{rudin2019stop, slack2020fooling}. The DIPW-Lasso estimates a structurally simple model for CATE, offering direct interpretability in line with arguments favoring inherently interpretable models over post hoc explanations.

\subsection{Plan of Paper}
The remainder of the paper is organized as follows. Section \ref{sec:methodology} introduces the setup and presents our denoising method, along with its connection to the AIPW transformation. Section \ref{sec:results} establishes the theoretical properties of the DIPW-Lasso, showing that it achieves a tighter error bound than its non-denoised counterpart, except on a set of asymptotically negligible probability. Section \ref{sec: simulation} reports results from a Monte Carlo study conducted in a setting with small treatment effects. Section \ref{sec:GOTV} presents empirical results based on two publicly available RCT datasets.
A conclusion is presented in Section \ref{sec:conclusion}. The technical proofs are presented in the Appendix.

\section{Methodology}
\label{sec:methodology}
This section presents our setting, an identification strategy, and an estimation method for CATE.

\subsection{Setup and Notation}
Let $Y_i$ represent the response outcome, $T_i\in \{0,1\}$ denote the binary treatment indicator, and $X_i=(X_{i1}, \dots, X_{ip})'$ be a $p$-dimensional vector of covariates, where $p$ can be larger than the sample size. 
We observe a random sample $\{Y_i,T_i,X_i\}^n_{i=1}$.

We formalize the causal problem within Rubin's potential outcome framework. 
Let $Y_i(1)$ and $Y_i(0)$ denote the potential outcomes when $T_i=1$ and $T_i=0$, respectively.
The CATE function is the key causal parameter that characterizes the heterogeneous treatment effects and is defined as
\begin{equation} 
\tau(x)=E[Y_i(1)-Y_i(0) \mid X_i=x]. 
\end{equation}

We assume that $\tau(x) = x'\beta$ for some $p$-dimensional potentially sparse vector $\beta$, while imposing no explicit assumptions on the nuisance functions $\mu_1(x) \equiv E[Y_i(1)|X_i=x]$ and $\mu_0(x) \equiv E[Y_i(0)|X_i =x]$.  
This assumption reflects our view that CATE is a relatively simple function, with only a small subset of covariates driving heterogeneity, whereas potential outcomes can depend on the covariates in a complex manner.
Some previous studies have indicated that the smoothness of $\tau(x)$ often differs substantially from that of $\mu_1(x)$ and $\mu_0(x)$  (see, e.g., \cite{kennedy2023towards}).

We impose the following assumptions for the identification of CATE:
\begin{assumption}\label{assum:unconf}
    The treatment is randomly assigned conditional on $X_i$:
        $$T_i \perp \!\!\! \perp (Y_i(1),Y_i(0)) \mid X_i.$$
\end{assumption}

\begin{assumption}\label{assum:overlap}
    The propensity score function $p(x)\equiv p(T_i=1 \mid X_i=x)$ is known and there exists $\xi > 0$ such that $\xi < p(x) < 1 - \xi$ for all $x$ in the support of $X_i$.
\end{assumption}

\subsection{IPW Transformation}

Under Assumptions \ref{assum:unconf} and~\ref{assum:overlap}, CATE can be identified as the conditional expectation of the IPW-transformed outcome:
\begin{equation}
   \tau(x) = E[Y_iW_i \mid X_i=x], \quad \text{where } W_i=\frac{T_i-p(X_i)}{p(X_i)(1-p(X_i))}.
\end{equation}
With a known propensity score, the IPW transformation generates an unbiased signal of CATE. 
This allows us to reformulate the CATE estimation problem as a simple regression task.
Specifically, one can estimate CATE by regressing the transformed outcome $Y_iW_i$ on the covariates $X_i$, which is often referred to as pseudo-outcome regression.
This method is suitable for our setting because we can directly impose a functional form on $\tau(x)$ without imposing any conditions on $\mu_1(x)$ and $\mu_0(x)$.
Under the linearity assumption on $\tau(x)$, we obtain the following linear regression model:
\begin{equation}\label{eq:non-denoised}
Y_i W_i = X_i'\beta + e_i, \quad E[e_i \mid X_i] =0,
\end{equation}
which can be estimated using the Lasso.

IPW and its related methods have been widely used for estimating ATE, and in recent years, they have also been employed for estimating CATE (\cite{curth2021nonparametric}; \cite{kennedy2023towards}).
However, IPW pseudo-outcome regression has been reported to perform poorly in some simulation settings (for example, \cite{knaus2021MonteCarloHTE}).
Further, previous observational studies have revealed that the instability of propensity score estimation can have undesirable effects.

The problem we observe is that the IPW transformation introduces significant noise even when the propensity score is known.
To illustrate this intuitively, we decompose the IPW-transformed outcome as follows:
\begin{equation} \label{IPWdecom1}
Y_i W_i = (Y_i(1) - Y_i(0)) T_i W_i + Y_i(0) W_i.
\end{equation}
Here, we have $E[Y_i(0) W_i|X_i] = 0$, indicating that the second term is pure noise. 
In practical applications, the baseline effect $Y_i(0)$ often exhibits greater variability than the causal effect $(Y_i(1) - Y_i(0))$, leading to a low signal-to-noise ratio. 
This issue undermines the accuracy of IPW pseudo-outcome regression.

\subsection{Denoising Method}

We introduce a method to reduce the noise in the IPW-transformed outcome. We begin by considering the following linear projection model, which augments \eqref{eq:non-denoised} with additional terms:
\begin{equation} \label{eq:denoised}
Y_i W_i = X_i'\beta + \alpha_1 W_i + \alpha_2 B(X_i) W_i + u_i,
\end{equation}
where $B(x)$ is a function of $X_i$, and $\alpha_1$ and $\alpha_2$ are linear projection coefficients.
Including the terms $W_i$ and $B(X_i)W_i$ does not affect the identification of CATE, because
\begin{equation} \label{orthogonality}
E \left[ B(X_i) W_i \tau(X_i) \right] = 0
\end{equation}
for any $B(x)$.
The orthogonality implies that the error terms in \eqref{eq:non-denoised} and \eqref{eq:denoised} are related as
\[
e_i = \alpha_1 W_i + \alpha_2 B(X_i) W_i + u_i,
\]
which shows that $E[u_i^2] \leq E[e_i^2]$.

Although the choice of $B(x)$ is unrestricted, the optimal function that minimize the variance of $u_i$ can be obtained by solving
\begin{equation}\label{eq:optimalB}
\min_{B(\cdot)} E[(Y_i W_i - B(X_i) W_i)^2].
\end{equation}
We prove the following proposition in the Appendix.
\begin{proposition}\label{proposition:optimalB}
The optimal function $B^*(x)$ that solves \eqref{eq:optimalB} is
$$B^*(x) = (1 - p(x)) \mu_1(x)+ p(x) \mu_0(x).$$
\end{proposition}

Proposition~\ref{proposition:optimalB} shows that our denoising transformation of the IPW outcome includes the AIPW transformation as a special case.
Using the optimal function, we obtain \[
    Y_i W_i-B^*(X_i) W_i = \frac{T_i(Y_i-\mu_1(X_i))}{p(X_i)}-\frac{(1-T_i)(Y_i-\mu_0(X_i))}{1-p(X_i)}+\mu_1(X_i) -\mu_0(X_i),
 \]
which corresponds exactly to the AIPW transformation, also known as the doubly robust transformation.
Proposition \ref{proposition:optimalB} provides a different perspective on the AIPW transformation; namely, it effectively reduces the noise introduced by the IPW transformation.

\subsection{DIPW-Lasso}\label{sec:Denoised IPW-Lasso}

We estimate the CATE coefficients \(\beta\) using the Lasso. 
To achieve this, we first determine the nuisance function $B(x)$. Proposition \ref{proposition:optimalB} shows that $B^*(x)$ is optimal if known. 
However, when replaced by its estimator, it may not necessarily retain its optimality .
Since $\mu_1(x)$ and $\mu_0(x)$ can be complex functions, $B^*(x)$ may be poorly estimated when the sample size is small or when the sample sizes of the treatment and control groups differ significantly.
Our simulation and empirical results show that $\mu(x) \equiv E[Y|X=x]$ performs reasonably well as $B(x)$ in many settings.
We employ machine learning methods such as random forest, boosting, and deep neural networks to estimate $B(x)$.

After determining the nuisance function, we estimate the linear projection model \eqref{eq:denoised} by using the cross-fitting procedure of \textcite{chernozhukov2018dml}.
The procedure is as follows (Algorithm 1).

\textit{Step 1:} Randomly split the sample into \(K\) non-overlapping folds of equal size, $\{I_k\}_{k=1}^K$. For each fold \(I_k\), estimate \(B(x)\) using the data in its complement, denoted by \(\widehat{B}_{I_k^c}(\cdot)\). Then, for each \(i \in I_k\), obtain the prediction \(\widehat{B}_{I^c_k}(X_i)\).

\textit{Step 2:} Solve the optimization problem
\begin{equation}\label{eq:object1}
    \min_{\alpha_1, \alpha_2, \beta} \left\{ \frac{1}{2n} \sum_{k=1}^K \sum_{i \in I_k} 
    \left(Y_iW_i - X_i'\beta -\alpha_1 W_i - \alpha_2 \widehat{B}_{I^c_k}(X_i) W_i \right)^2 + \lambda \|\beta\|_1 \right\},
\end{equation}
where $\lambda$ is the regularization parameter.

We also consider an alternative algorithm to estimate $\beta$ (Algorithm 2).
Given that $\widehat B_{I^c_k}(\cdot)$ is estimated using the same procedure as in Step 1, in Step 2, we first estimate $\alpha_1$ and $\alpha_2$ by regressing $Y_i W_i$ on $(W_i,\widehat B_{I^c_k}(X_i)W_i)'$ and obtain the OLS residual:
\begin{equation}\label{eq:OLSresidual}
    \widehat{Y}^*_i = Y_iW_i - \widehat{\alpha}_1W_i - \widehat{\alpha}_2\widehat{B}_{I^c_k}(X_i) W_i.
\end{equation}
We then solve the following optimization problem:
\begin{equation}\label{eq:objec2}
    \min_{\beta} \left\{ \frac{1}{2n} \sum_{i=1}^n \left(\widehat{Y}^*_i - X'_i \beta\right)^2 + \lambda \|\beta\|_1 \right\}.
\end{equation}

The two algorithms are asymptotically equivalent.
This equivalence follows from two observations.
First, if $X_i$ is also replaced with the OLS residual from regressing it on $(W_i,\widehat B_{I^c_k}(X_i)W_i)'$ in Algorithm 2, the resulting estimator for $\beta$ is identical to the one obtained by solving \eqref{eq:object1} by the Frisch-Waugh-Lovell theorem for the Lasso (\cite{yamada2017frisch}).
Second, due to the orthogonality \eqref{orthogonality}, $X_i$ is asymptotically equivalent to its residual.
Consequently, \eqref{eq:object1} and \eqref{eq:objec2} are asymptotically equivalent.

Algorithm 2 can be easily extended to the case where no functional form is assumed for $\tau(x)$.
Using $\widehat Y_i^*$ as the pseudo-outcome, any nonparametric method can be applied to estimate CATE.

Another advantage of using Algorithm 2 is that it enables an evaluation of the effectiveness of the denoising procedure. Specifically, the \(R^2\) from regressing \(Y_i W_i\) on \((W_i, \widehat B_{I^c_k}(X_i) W_i)'\) provides a rough measure of the extent of denoising, where an $R^2$ value close to zero indicates that the denoising is not effective.

\section{Theoretical Results}
\label{sec:results}

We present the theoretical properties of the DIPW-Lasso based on Algorithm 2.
It is evident that denoising does not improve the convergence rate of the Lasso.
Therefore, we compare finite-sample prediction error bounds for the non-denoised and denoised Lasso.
We show that the DIPW-Lasso achieves a better error bound, even when $(\alpha_1, \alpha_2)$ and $B(x)$ are estimated, except for an event with an asymptotically negligible probability.

\subsection{Key Inequality for the Lasso}\label{sec:keyInequality}

We briefly review the finite-sample properties of the Lasso to understand how our denoising method improves a prediction error bound.

Let $\tilde{\beta}$ be the Lasso estimator for \eqref{eq:non-denoised} with the regularization parameter $\lambda$.
Further, let $X=(X_1, \dots, X_n)'$ and $e=(e_1, \dots, e_n)'$.
Using the so-called basic inequality (\cite{buhlmann2011statistics}) and the H\"older inequality, we obtain
\begin{equation}\label{eq:BasicInequality}
\frac{1}{n}\|X \tilde{\beta} - X \beta\|_2^2 \leq 2 \left\|\frac{1}{n} X'e \right\|_{\infty} \|\tilde{\beta} - \beta\|_1 + 2 \lambda\left(\|\beta\|_1 - \|\tilde{\beta}\|_1\right).
\end{equation}
Next, we define the event $\mathcal{T} = \left\{ \left\| X'e/n \right\|_{\infty} \leq \lambda_0 \right\}$.
Then, for \(\lambda \geq \lambda_0\), the following inequality holds under \(\mathcal{T}\):
\begin{equation}\label{eq:lasso_bound}
\frac{1}{n} \|X \tilde{\beta} - X\beta\|_2^2 \leq 4 \lambda \|\beta\|_1.
\end{equation}

The above result shows that the performance of the Lasso critically depends on $X'e/n$.
There exist known upper bounds for its $l_\infty$-norm, and $\lambda_0$ typically depends multiplicatively on the standard deviation of $e_i$  (see, e.g., \cite{buhlmann2011statistics}).
Thus, if the variance of the error term is reduced, then the prediction error bound can be improved by choosing a smaller $\lambda$.

\subsection{Prediction Error Bound}
We derive the prediction error bound for the IPW-Lasso and DIPW-Lasso and show that the denoising method improves this bound.

We impose the following two assumptions to derive the prediction error bound for the non-denoised IPW Lasso.

\begin{assumption}\label{assum:BoundednessofX}
There exists a positive constant $C$ such that $\max_{1 \leq j \leq p} \|X^j\|_\infty \leq C < \infty$, where $X^{j}$ is the $j$-th column of $X$.
\end{assumption}

\begin{assumption}\label{assum:finiteL4moment}
    The error term $e_i$ has a finite second moment.
\end{assumption}

Assumption \ref{assum:BoundednessofX} is adopted  from Assumption (A2) in  \textcite{buhlmann2015high} (see also Section 2.3.4 of \cite{van2014asymptotically} ). 
The boundedness assumption is sometimes employed in the random design setting to establish concentration inequalities.
Note that we do not assume independence between $X_i$ and $e_i$, as denoising would not be effective if they were independent.

Under Assumptions \ref{assum:BoundednessofX} and \ref{assum:finiteL4moment}, we can show that for any $\eta \in (0,1)$,
\begin{equation}\label{eq:noise_concentration}
P\left( \left\| \frac{1}{n} X'e \right\|_{\infty} \leq \frac{\sqrt{8} C\sigma_e \sqrt{\frac{\log 2p}{n}}}{\eta} \right) \geq 1 - \eta,
\end{equation}  
where $\sigma_e = \sqrt{E[e_i^2]}$.
The same inequality also holds when replacing $e$ and $\sigma_e$ with $u=(u_1, \dots, u_n)'$ and $\sigma_u =\sqrt{E[u_i^2]}$, respectively.

It follows from \eqref{eq:BasicInequality} and \eqref{eq:noise_concentration} that the non-denoised IPW-Lasso has the following property.

\begin{proposition}\label{proposition: basic_lassobound}
Suppose that Assumptions \ref{assum:BoundednessofX} and \ref{assum:finiteL4moment} hold. Let \(\lambda \geq \frac{ \sqrt{8} C\sigma_e}{\eta} \sqrt{\frac{\log 2p}{n}}\) for \(\eta \in (0,1)\). Then, we have
\[
\frac{1}{n} \| X \tilde{\beta} - X \beta \|_2^2 \leq 4\lambda \| \beta \|_1
\]
with probability at least \(1 -\eta\).  
\end{proposition}

Next, we investigate the properties of the DIPW-Lasso using Algorithm 2.
We first impose the following assumption on $\widehat B_{I_k^c}(x)$.

\begin{assumption}\label{assum:l2convergence}  
For $k=1, \dots, K$, \(\widehat{B}_{I_k^c}(x)\) converges in $L_2$ to a function $B(x)$:
\[
\lim_{n \to \infty} E \left[\int \left(\widehat{B}_{I_k^c}(x) - B(x)\right)^2 dF(x) \right] = 0, 
\]
where $F(x)$ is the distribution function of $X_i$.
\end{assumption} 

Assumption~\ref{assum:l2convergence} only requires $\widehat{B}_{I_k^c}(x)$ to converge to a certain function, whereas
we do not impose any requirement on its convergence rate.
This allows for flexible estimation in conjunction with a wide range of machine learning methods.

For the function $B(x)$ given in Assumption~\ref{assum:l2convergence}, we define the population counterpart of $\hat{Y}_i^*$ as
\begin{equation}
  Y_i^*=Y_i W_i - Z_i'\gamma, 
\end{equation}
where
$Z_i = (W_i, B(X_i)W_i)'$ and $\gamma = (\alpha_1, \alpha_2)'$. 
To ensure that $\gamma$ is well-defined, we impose the following assumption.
\begin{assumption}\label{assum:NonSingularity}  
$\det\left(E\left[ Z_i Z_i' \right]\right) \neq 0$.
\end{assumption}

An inequality similar to \eqref{eq:BasicInequality} can be derived for the estimator $\hat \beta$ that solves \eqref{eq:objec2}.
Let $\hat Z_i=(W_i, \hat{B}_{I_k^c}(X_i)W_i)$,
and define $Z=(Z_1, \dots, Z_n)'$, $\hat{Z}=(\hat Z_1, \dots, \hat Z_n)'$, and $Y^* = (Y_1^*, \dots, Y_n^*)'$.
Then, we obtain
\begin{equation}\label{eq:BasicInequalityofDIL}
\begin{aligned}
    \frac{1}{n}\|X \hat{\beta} - X \beta\|_2^2 & \leq  2 \left\|\frac{1}{n} X' \left(u + \left(Z - P_{\hat{Z}} Z\right) \gamma - P_{\hat{Z}}Y^* \right)\right\|_{\infty}\|\hat{\beta} - \beta\|_1  \\
    & \quad +2 \lambda\left(\|\beta\|_1 - \|\hat{\beta}\|_1\right),
\end{aligned}
\end{equation}
where \(P_{\hat{Z}} = \hat{Z} (\hat{Z}' \hat{Z})^{-1} \hat{Z}'\).
The effects of estimating $B(\cdot)$ and $\gamma$ are captured by the terms $X'\left(Z - P_{\hat{Z}} Z\right)\gamma/n$ and $X'P_{\hat{Z}}Y^*/n$, respectively.
Under Assumptions \ref{assum:l2convergence} and \ref{assum:NonSingularity}, these two terms are of smaller order than $X'u/n$ in terms of the $\ell_\infty$ norm.

The following theorem for the DIPW-Lasso can now be established.
\begin{theorem}\label{main_theorem}
Suppose that Assumptions~\ref{assum:overlap} and \ref{assum:BoundednessofX}--\ref{assum:NonSingularity} hold.
Let \(\lambda \geq (\frac{\sqrt{8} C\sigma_u}{\eta} + \epsilon) \sqrt{\frac{\log 2p}{n}}\) for \(\eta \in (0,1)\) and \(\epsilon > 0\). Then, we have
\[
\frac{1}{n} \| X \hat{\beta} - X \beta \|_2^2 \leq 4 \lambda \| \beta \|_1,
\]
with probability at least \(1 - \eta- \delta_n\),
where
\[
\delta_n = P \left( \left\| \frac{1}{n} X'(Z-P_{\hat Z}Z)\gamma + \frac{1}{n} X'P_{\hat Z}Y^* \right\|_\infty \geq \epsilon \sqrt{\frac{\log 2p}{n}} \right).
\]
Moreover, $\delta_n \to 0$ as $n \to \infty$.
\end{theorem}  

Theorem \ref{main_theorem} states that the prediction error bound can be improved by denoising as long as $\sigma_u < \sigma_e$, although the probability that the bound holds decreases by an asymptotically negligible amount.

Although Theorem~\ref{main_theorem} gives the error bound for the DIPW-Lasso when $\lambda$ is chosen according to a theoretically justified rule, we recommend choosing $\lambda$ via cross-validation in practice, because this
often leads to better performance than more theoretically motivated choices. Although it is desirable to analyze the properties of the DIPW-Lasso when $\lambda
$ is selected through cross-validation, we avoid this here due to its technical complexity (see, e.g., \cite{chetverikov2021cross} for the properties of the cross-validated Lasso).
Instead, we confirm through simulations that the regularization parameter selected by cross-validation becomes significantly smaller when denoising is applied compared to when it is not.

\begin{remark}
Theorem \ref{main_theorem} holds even if $\beta$ is not sparse.
If, however, $\beta$ is sparse, then the prediction error rate can be improved by adding an assumption on $X$.
For instance, suppose that the compatibility condition is satisfied almost surely with compatibility constant $\phi_0 >0$ (see \cite{buhlmann2011statistics} for the definition of the compatibility condition).
Then, plugging in $\lambda= 2\left(\frac{ \sqrt{8} C\sigma_u}{\eta} + \epsilon \right) \sqrt{\log 2p/n}$, with probability at least $1-\eta - \delta_n$, we obtain
\[
\frac{1}{n} \| X \hat{\beta} - X \beta \|_2^2 \leq  \frac{36 \left(\frac{\sqrt{8}C \sigma_u }{\eta} + \epsilon  \right)^2 s_0 \log 2p}{n \phi_0^2},
\]
where $s_0$ is the number of nonzero elements in $\beta$ (see Theorem 6.1 of \cite{buhlmann2011statistics}).
We can confirm that an improvement in the bound is possible compared to the case without denoising.

\end{remark}

\section{Monte Carlo Simulation}\label{sec: simulation}
In this section, we evaluate the improvement of the DIPW-Lasso over the non-denoised version through simulation studies. Furthermore, we compare our method with recent machine learning-based estimation methods for the CATE, namely, the causal forest \parencite{wager2018causalforest}, T-learner, X-learner \parencite{kunzel2019metalearners}, R-learner \parencite{nie2021quasi}, and DR-learner \parencite{kennedy2023towards}.

\subsection{Simulation designs}\label{sec: simudesign}
We conducted our experiments based on  the design of \textcite{nie2021quasi}.
We generate 50 covariates, where \( X_{i1} \) through \( X_{i30} \) are independently drawn from \(\text{Unif}(0, 1) \), and \( X_{i31} \) through \( X_{i50}\) are independently drawn from \(\mathcal{N}(0, 1) \).
We construct complex potential outcome functions and a sparse linear CATE model as follows:
\begin{align*}
Y_i &= b(X_i) + T_i \tau(X_i) + \epsilon_i \\
b(X_i) &= 5 \big\{ \sin(\pi X_{i1} X_{i2}) + 2\bigl(X_{i3} - 0.5\bigr)^2 + X_{i4} + 0.5\,X_{i5} 
      + 2\log\bigl(1 + \exp(X_{i31} + X_{i32} + X_{i33})\bigr) \\
    &\quad + \max \{0, X_{i31} + X_{i32} + X_{i33} \}
      + \max \{ 0, X_{i34} + X_{i35} \} \big\}, \\
\tau(X_i) &= 0.5\bigl(X_{i1} + X_{i2}\bigr) + X_{i4} + \tfrac{1}{3}X_{i32} + 2\,X_{i40},
\end{align*}
where \( \epsilon_i \mid X_i\sim \mathcal{N}(0, 1) \). 
The data-generating process simulates a challenging weak signal setting where only approximately 3.5\% of the outcome variation can be attributed to the treatment effect.

The treatment indicator \(T_i\) follows a Bernoulli distribution with a known probability \(p=P(T_i=1)\). We consider \(p = 0.5\) and \(p = 0.2\) in our simulations. The latter case represents a treatment-control imbalance setting.
Experimental settings with imbalanced treatment and control groups are common in practice. For example, in marketing and advertising experiments, it is rare to assign half of the units to treatment or control due to cost considerations. Therefore, the performance of CATE estimation in such imbalanced settings is of practical importance.

We report the results for  the DIPW-Lasso using Algorithm 1, with $B(x) =\mu(x)$ and $K=5$.
To enable a fair performance comparison, all nuisance parameters are estimated using random forests with the same configuration across all methods.
Moreover, we employ the Lasso for estimating CATE in all pseudo-outcome regression methods (i.e., the R-learner, DR-learner, IPW, and DIPW), with the regularization parameter selected via 10-fold cross-validation.
For the X-learner, we also apply the Lasso in the CATE estimation stage to take advantage of the structural information.
For methods where it is difficult to impose linearity on CATE, namely the causal forest and T-learner, we follow their standard implementation procedures.
Therefore, the experimental setting is somewhat disadvantageous for the causal forest and T-learner.

\subsection{Simulation results}
Here, we present the simulation results and briefly discuss them.
We generate \( n =1000 \) observations for training and \( 10000\) independent observations for evaluating the out-of-sample RMSE. For each method, the RMSE is calculated 500 times through repeated simulations of CATE estimation.
%\newpage
\begin{figure}[htbp]
  \centering
  % Left image subfigure
  \begin{subfigure}[t]{0.49\textwidth}
    \centering
    \includegraphics[width=\linewidth]{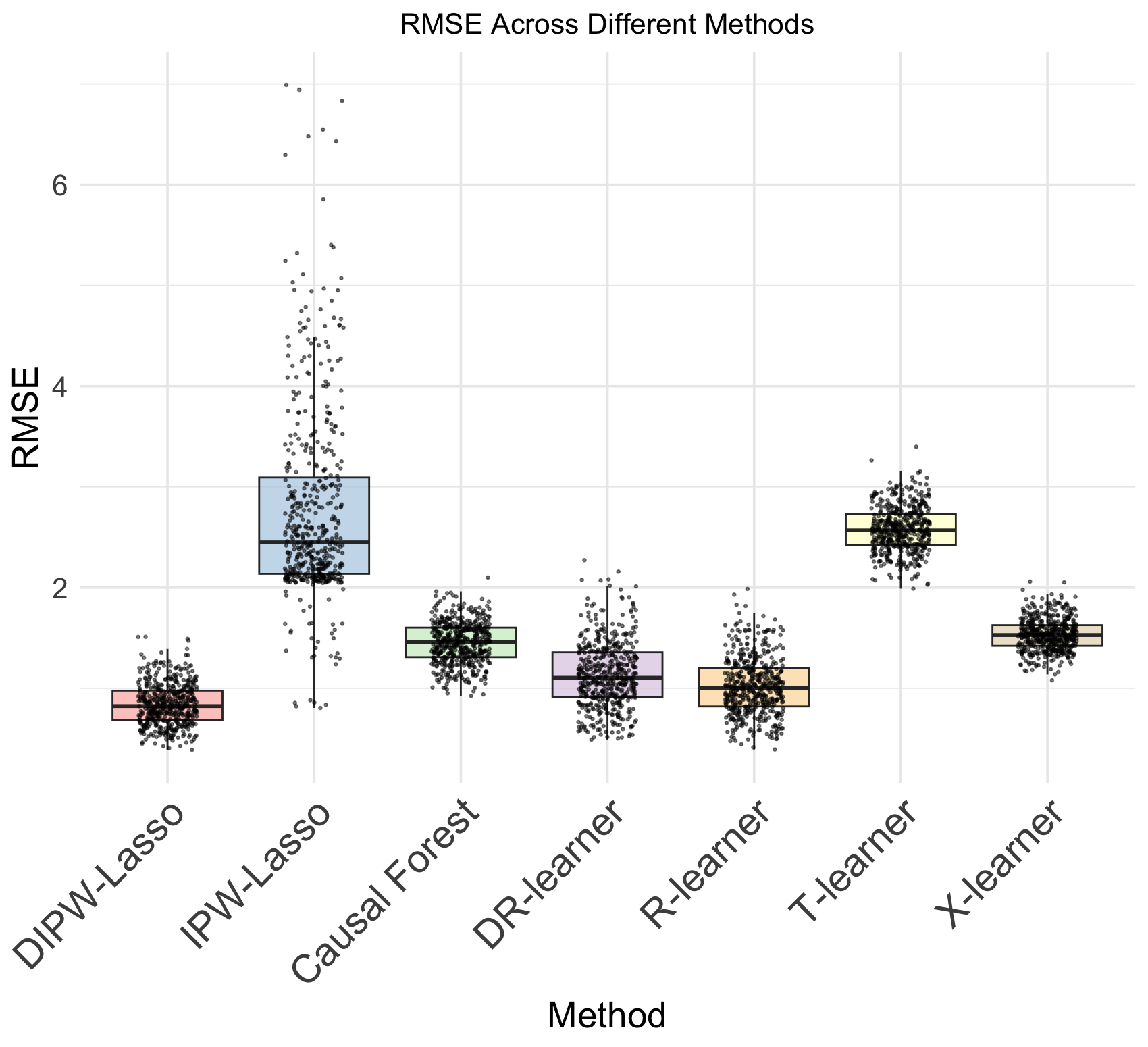}
    \caption{RMSE}
    \label{fig:0.5rmse}
    \par\smallskip
  \end{subfigure}
  \hfill
  % Right image subfigure
  \begin{subfigure}[t]{0.49\textwidth}
    \centering
    \includegraphics[width=\linewidth]{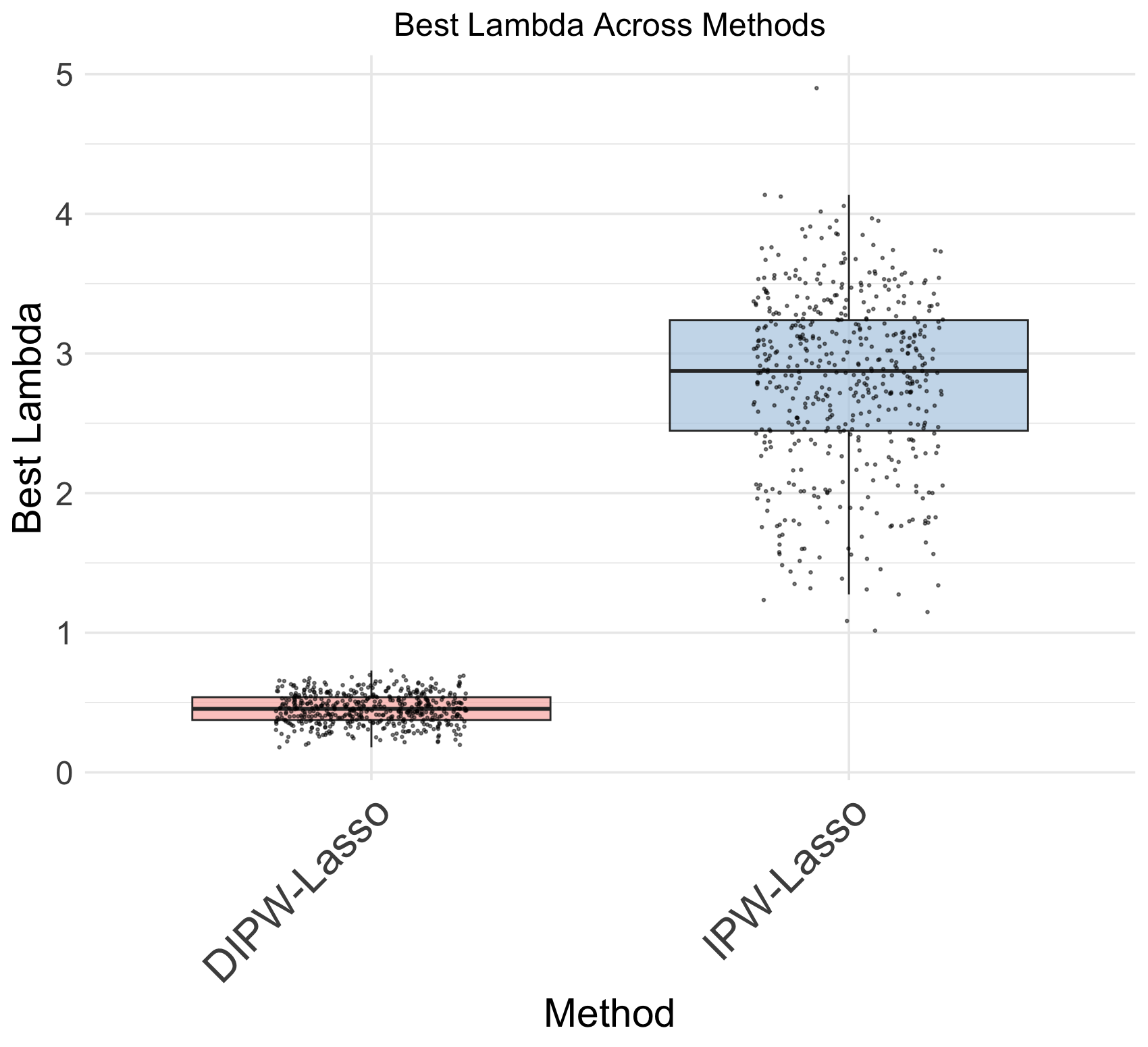} 
    \caption{$\lambda$ selection }
    \label{fig:0.5lambda}
    \par\smallskip
  \end{subfigure}
  \caption{Simulation results for \(p = 0.5\). Boxplots for (a) RMSE results for CATE estimation and (b) 10-fold cross-validation results for \(\lambda\) in DIPW-Lasso and IPW-Lasso. } 
  \label{fig:0.5simuresult}
\end{figure}

Figure \ref{fig:0.5rmse} presents the result for $p=0.5$, showing that the DIPW-Lasso achieved the best performance among all methods. As we would expect,  the two methods that do not allow linearity to be imposed on CATE tend to perform poorly.
Incorporating structural information about CATE substantially improves estimation accuracy, although the X-learner performs rather poorly.
Moreover, we find that the DIPW-Lasso dramatically outperforms the IPW-Lasso. 
The IPW-Lasso not only has a larger average RMSE but also exhibits greater variance, indicating that its estimation results are quite unstable.

Figure \ref{fig:0.5lambda} plots the values of the regularization parameters for the IPW-Lasso and DIPW-Lasso selected by 10-fold cross-validation.
As shown, the DIPW-Lasso exhibits greater stability in the selected regularization parameters, achieving both a smaller average and a smaller variance in comparison to the non-denoised version.
This is consistent with our theoretical result that improved estimation performance is associated with the selection of a smaller regularization parameter.
%0.2
\begin{figure}[htbp]
  \centering
  % Left image subfigure
  \begin{subfigure}[t]{0.48\textwidth}
    \centering
    \includegraphics[width=\linewidth]{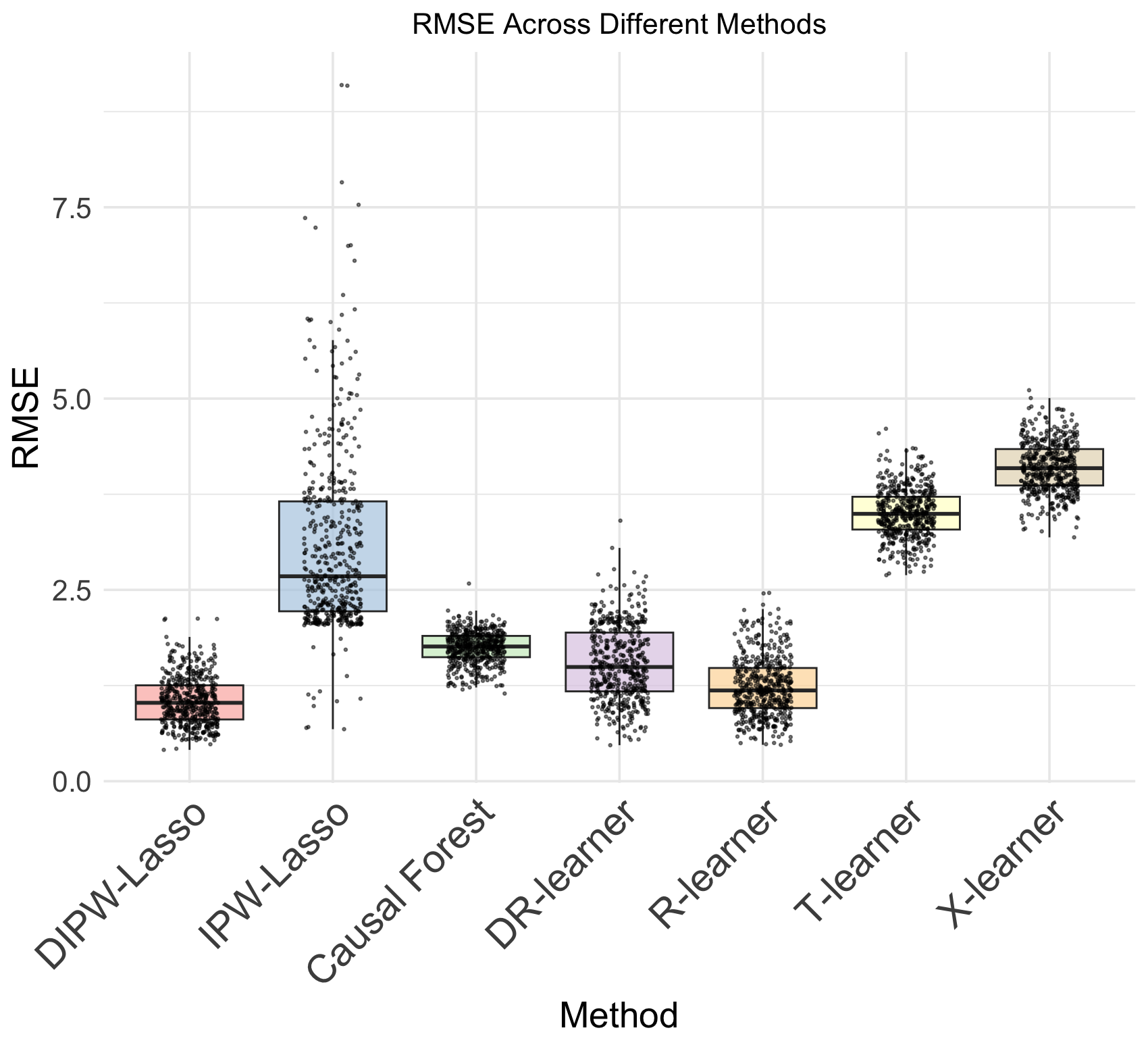}
    \caption{RMSE}
    \label{fig:0.2rmse}
    \par\smallskip
  \end{subfigure}
  \hfill
  % Right image subfigure
  \begin{subfigure}[t]{0.48\textwidth}
    \centering
    \includegraphics[width=\linewidth]{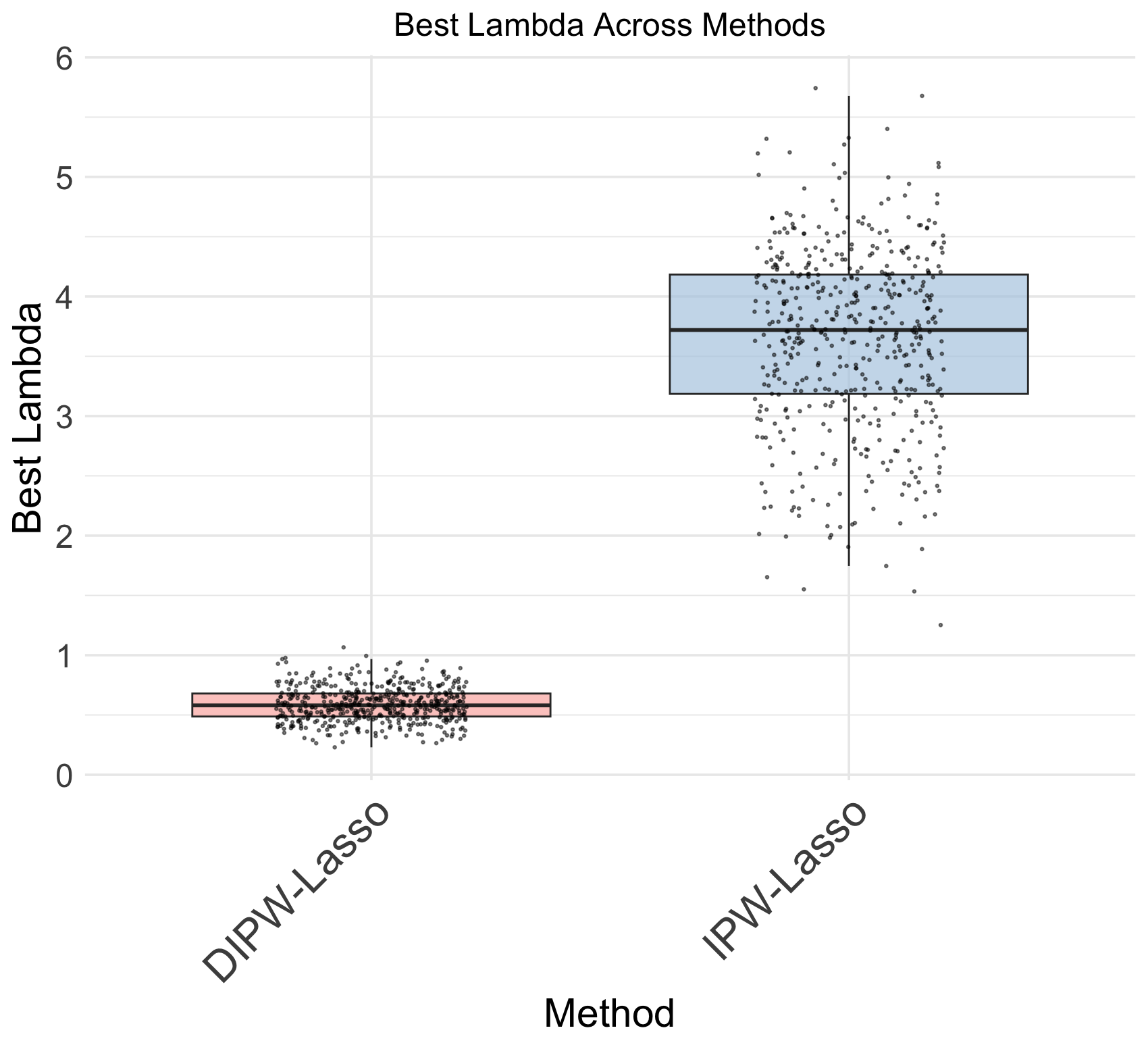} 
    \caption{$\lambda$ selection}
    \label{fig:0.2lambda}
    \par\smallskip
  \end{subfigure}
  \caption{Simulation results for \(p = 0.2\). Boxplots of (a) RMSE results for CATE estimation and(b) 10-fold cross-validation results for \(\lambda\) in DIPW-Lasso and IPW-Lasso.}
  \label{fig:0.2simuresult}
\end{figure}

Figure~\ref{fig:0.2rmse} shows the results for $p=0.2$, in which the DIPW-Lasso again achieves the highest accuracy. The DR-learner is less competitive in this setting, as the DR-learner uses both \(\mu_1 \) and \(\mu_0\) as nuisance functions, and when the treatment group is small, the estimation of \( \mu_1\) becomes less accurate, leading to a decline in overall performance.

Note, although the results are not reported in detail in the present paper, we also conducted experiments with other values replacing the constant 5 multiplier in the equation for $b(X_i)$.
As this multiplier decreases, the performance gap between the DIPW-Lasso and other methods narrows, but the DIPW-Lasso still achieves the lowest RMSE in all cases.

\subsection{Evaluating CATE Estimation via Ranking Ability}\label{sec:upliftmodeling}
Another method for evaluating the performance of CATE estimation is to assess its ranking ability, commonly referred to as uplift modeling. Uplift modeling was initially introduced in marketing analytics (e.g., \cite{radclifte2008identifying}) and has recently been incorporated into causal inference, with applications in fields such as personalized medicine and policy evaluation (e.g., \cite{jaskowski2012uplift,athey2025machine}).
For a comprehensive review of uplift modeling techniques, see, for example, \textcite{gutierrez2017causal}.

We provide only a brief summary of uplift modeling here.
Given an estimated CATE model $\hat{\tau}(x)$, we rank the units in the test set in descending order of their predicted CATE values. 
The uplift value for the top-$k$ units is computed as
\begin{equation}\label{eq:upliftvalue}
U(k)=\widehat{ATE}_k \left(\hat{\tau}(x) \right) \times k,
\end{equation}
where $\widehat{ATE}_k \left(\hat{\tau}(x) \right)$ denotes the difference-in-means estimator of ATE among the top-$k$ units.
The graph of $U(k)$ is called the uplift curve\footnote{There are several ways to compute the uplift curve (see, e.g., \cite{devriendt2020learning}) and a variant known as the Qini curve. Formula \eqref{eq:upliftvalue} is a commonly used method in which the uplift value directly reflects the expected gain from treating the top-ranked units. This approach also aligns  with the scikit-uplift package (https://www.uplift-modeling.com/en/latest/) and  formula (8) in \textcite{devriendt2020learning}.}. Figure \ref{fig:upliftcurveexample} presents uplift curves obtained from the IPW-Lasso and DIPW-Lasso. The ranking ability of each method is evaluated by the area under the uplift curve (AUUC), where a higher AUUC reflects better ranking performance.

\begin{figure}[htbp]
    \centering
    \includegraphics[width=0.6\linewidth]{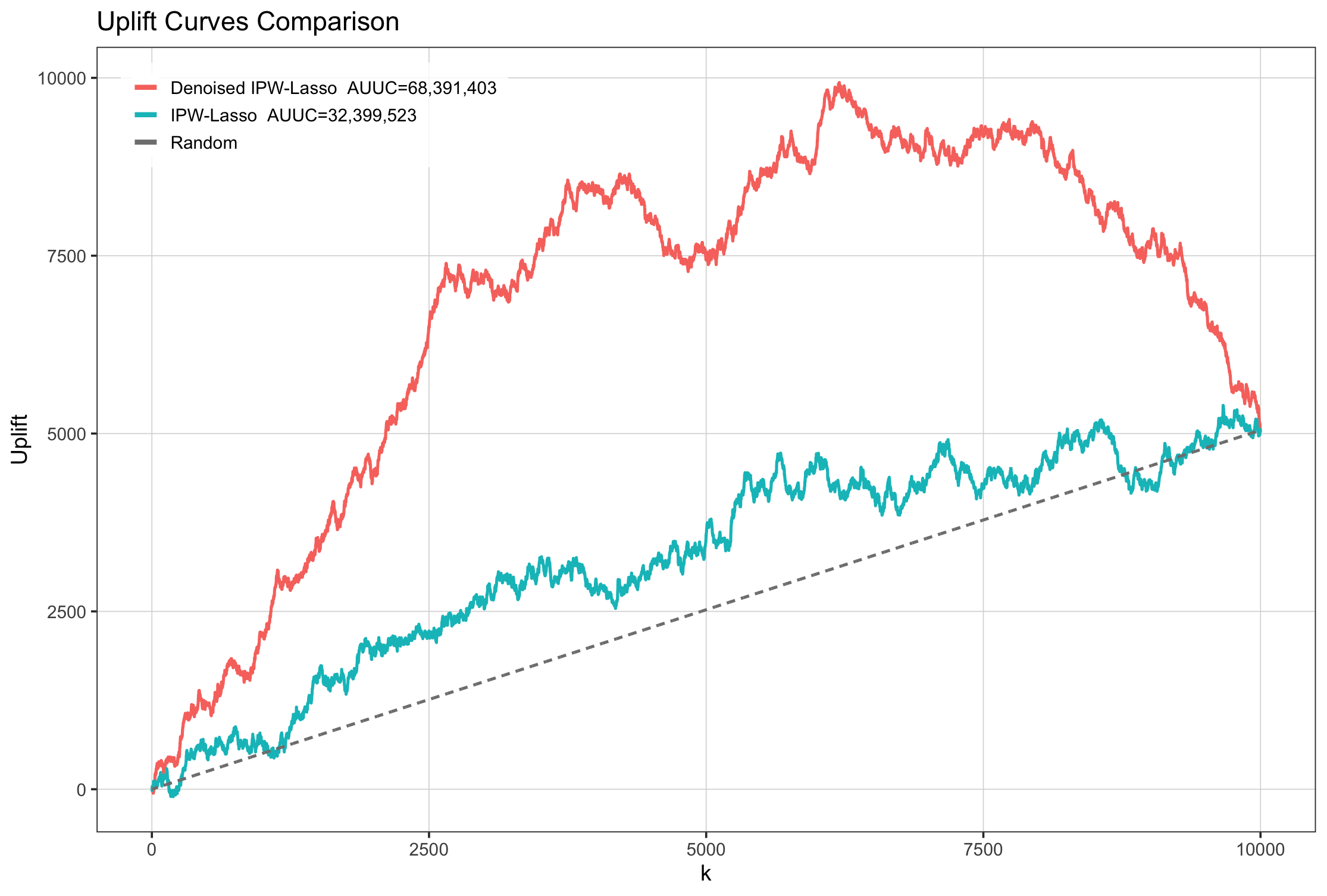}
    \caption{Example of uplift curve based on simulated data $p=0.5$, as introduced in Section \ref{sec: simudesign}. The result shows that the DIPW-Lasso achieved a higher AUUC than the non-denoised version. The improvement in CATE estimation by the proposed denoising method also gives rise to greater ranking performance.}
    \label{fig:upliftcurveexample}
\end{figure}

To examine whether the AUUC reflects the performances of different CATE estimations , we compute the AUUC for each method using the results in Section \ref{sec: simudesign}.

\begin{figure}[htbp]
  \centering
  %---------------- left：p = 0.5 ----------------
  \begin{subfigure}[t]{0.48\textwidth}
    \centering
    \includegraphics[width=\linewidth]{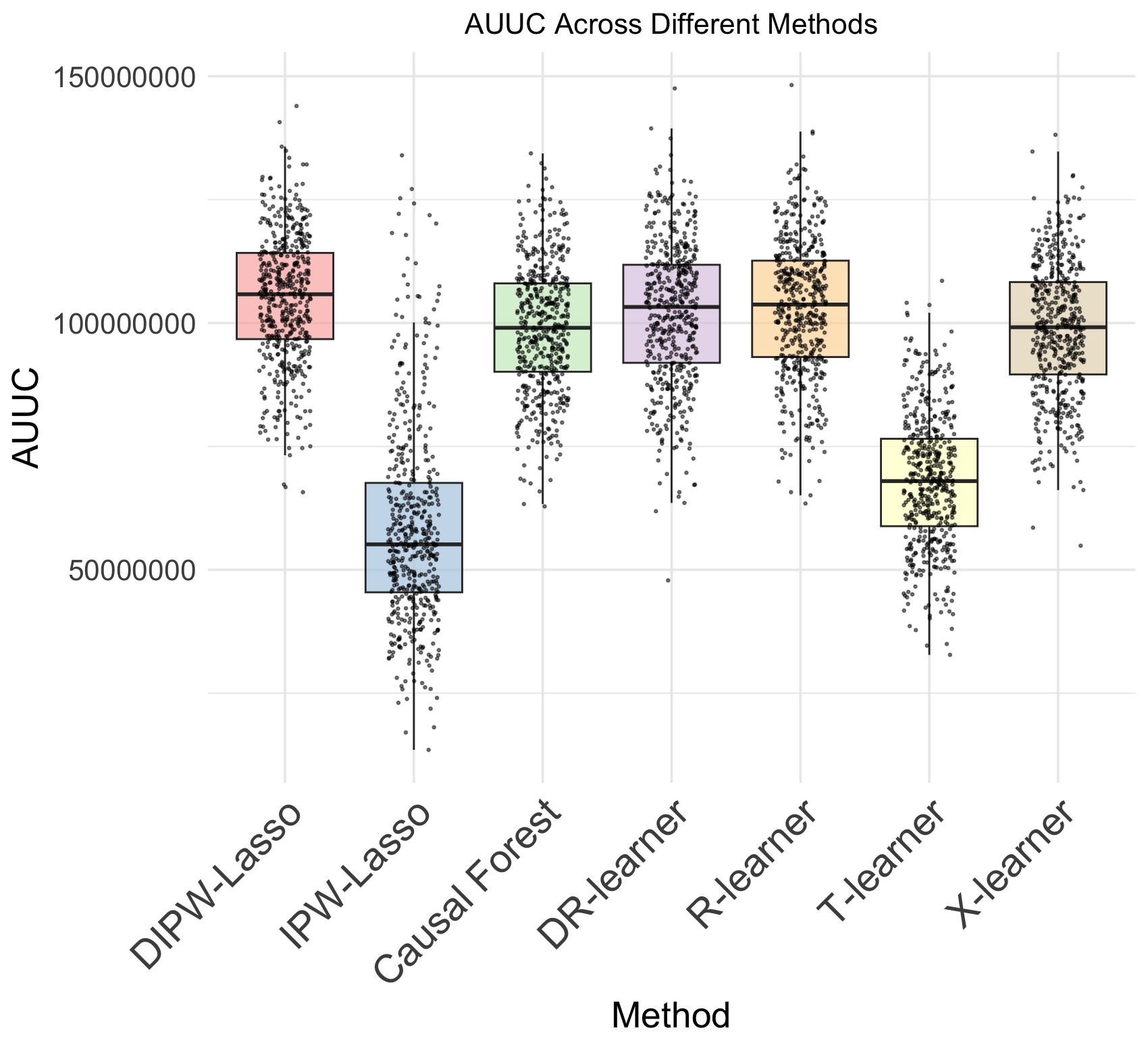}
    \caption{AUUC, \(p = 0.5\)}
    \label{fig:auuc_p0.5}
  \end{subfigure}
  \hfill
  %---------------- right：p = 0.2 ----------------
  \begin{subfigure}[t]{0.48\textwidth}
    \centering
    \includegraphics[width=\linewidth]{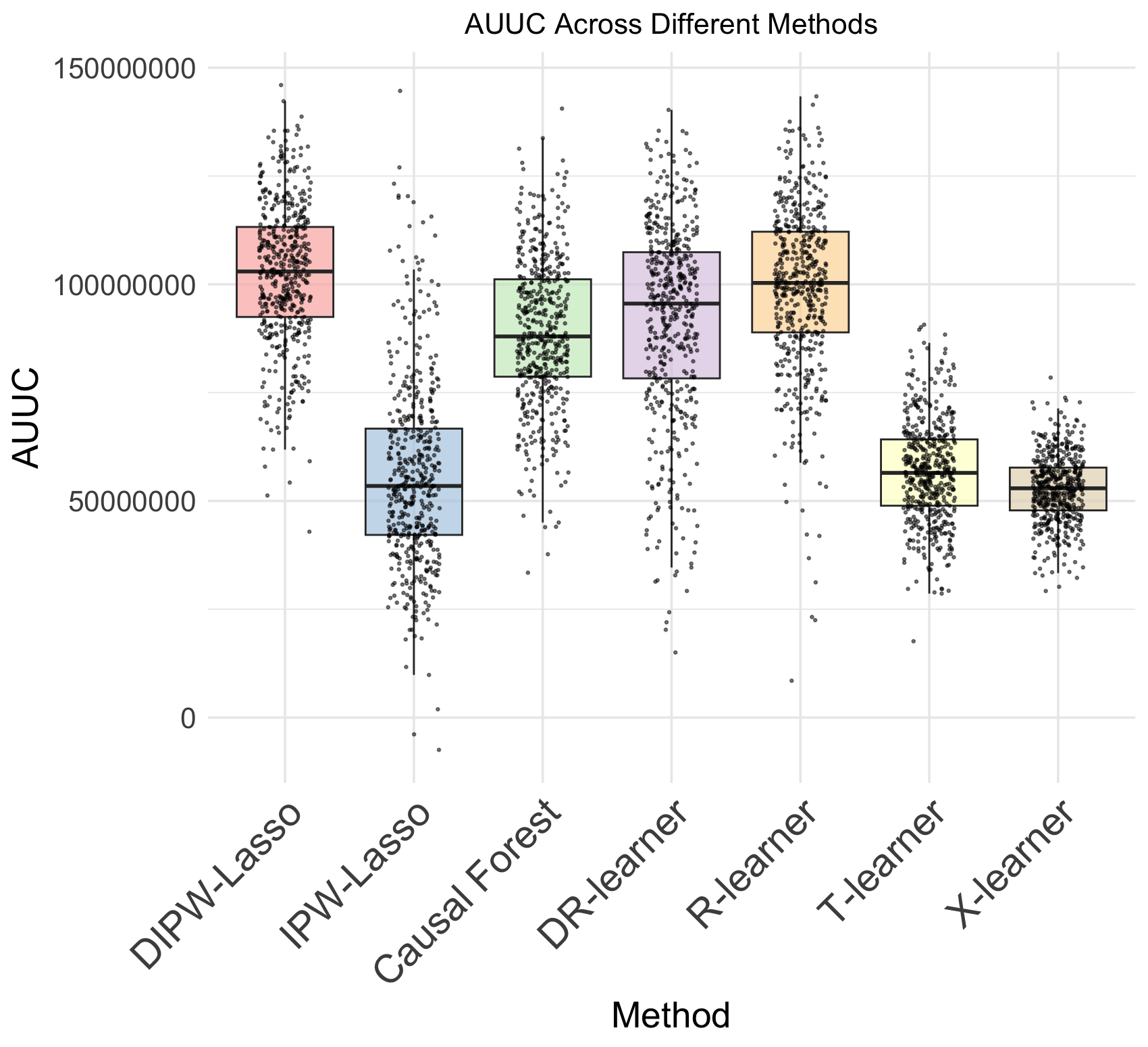}
    \caption{AUUC, \(p = 0.2\)}
    \label{fig:auuc_p0.2}
  \end{subfigure}

  \caption{Boxplots of AUUC over 500 simulation runs.}
  \label{fig:auuc_boxplots}
\end{figure}

\begin{table}[htbp]
  \centering
  %---------------- left：p = 0.5 ----------------
    \begin{subtable}[t]{0.48\textwidth}
      \centering
      \begin{adjustbox}{max width=\linewidth}
        \begin{tabular}{lrr}
          \toprule
          \textbf{Method} & \textbf{AUUC} & \textbf{RMSE} \\
          \midrule
          DIPW-Lasso     & 104{,}833{,}786 & 0.84 \\
          R-learner      & 103{,}172{,}229 & 1.02 \\
          DR-learner     & 102{,}009{,}295 & 1.13 \\
          Causal-Forest  & 98{,}885{,}547  & 1.45 \\
          X-learner      & 98{,}775{,}994  & 1.53 \\
          T-learner      & 67{,}991{,}024  & 2.58 \\
          IPW-Lasso      & 58{,}641{,}228  & 2.76 \\
          \bottomrule
        \end{tabular}
      \end{adjustbox}
      \caption{\(p = 0.5\)}
      \label{tab:metrics_p0.5}
    \end{subtable}
  \hfill
  %---------------- right：p = 0.2 ----------------
    \begin{subtable}[t]{0.48\textwidth}
      \centering
      \begin{adjustbox}{max width=\linewidth}
        \begin{tabular}{lrr}
          \toprule
          \textbf{Method} & \textbf{AUUC} & \textbf{RMSE} \\
          \midrule
          DIPW-Lasso     & 102{,}494{,}879 & 1.05 \\
          R-learner      & 99{,}121{,}149  & 1.23 \\
          DR-learner     & 91{,}780{,}034  & 1.54 \\
          Causal-Forest  & 89{,}006{,}750  & 1.75 \\
          T-learner      & 56{,}743{,}482  & 3.51 \\
          IPW-Lasso      & 55{,}989{,}419  & 3.07 \\
          X-learner      & 52{,}791{,}721  & 4.10 \\
          \bottomrule
        \end{tabular}
      \end{adjustbox}
      \caption{\(p = 0.2\)}
      \label{tab:metrics_p0.2}
\end{subtable}
  \caption{Mean AUUC and RMSE across 500 simulation runs for each method.}
  \label{tab:metrics_all}
\end{table}

The simulation results shown in Figure~\ref{fig:auuc_boxplots} and Table~\ref{tab:metrics_all} suggest that better ranking performance tends to be associated with more accurate CATE estimation. The DIPW-Lasso performs best in terms of both AUUC and RMSE. Overall , uplift modeling serves as a valuable complement to RMSE-based evaluation. It is particularly advantageous in empirical studies because of it not requiring knowledge of the true causal effects, which are unobservable in practice.

\section{Real-World Applications}\label{sec:GOTV}
In this section, we revisit two large-scale experimental datasets, the GOTV dataset and the Criteo Uplift Modeling dataset, to show the efficacy of the DIPW-Lasso.

To evaluate performance, we randomly split each dataset, assigning 75\% as the training set and the remaining 25\% as the test set. We estimate the CATE on the training set using six methods: the DIPW-Lasso, IPW-Lasso (without denoising), T-learner, X-learner, DR-learner, and causal forest. For the T-learner and X-learner, we use random forests as base learners. Random forests are also used to estimate the nuisance components in the DIPW-Lasso and DR-learner. In the pseudo-outcome regression step of the DIPW-Lasso, IPW-Lasso, and DR-learner, we apply Lasso regression. For the causal forest, we implement the method using the \texttt{econml} library. Finally, we evaluate the performance of each method on the test set.

\subsection{Get-Out-The-Vote experiment}
We use data from the Social Pressure and Voter Turnout experiment by \textcite{gerber2008social}, a large-scale field experiment conducted in Michigan during the 2006 general election. The experiment was designed to evaluate the effects of four social pressure treatments on voter mobilization. Among these treatments, ``neighbors'', which involved sending mailers that promised to publicize recipients' own turnout records to their neighbors after the election, has attracted considerable attention due to its strong effect. As reported by \textcite{gerber2008social}, the ``neighbors'' treatment increased voter turnout by 8.1 percentage points on average. 

Our empirical analysis has two goals: (1) to estimate CATE and leverage the model to optimize treatment allocation; and (2) to uncover interpretable patterns of treatment effect heterogeneity.

Our data consist of individuals from the ``neighbors'' treatment group (38,201 observations) and the control group (191,243 observations). Binary voting records from the five elections held prior to the 2004 general election were used to construct a count variable representing the past voting history. The count variable was then one-hot encoded with zero voting history as the reference category. Other demographic variables include the gender indicator, age, and household size. The outcome is a binary indicator of turnout in the 2006 local elections.

\subsubsection{Targeting Performance }

%\newpage
\begin{figure}[htbp]
    \centering
    \includegraphics[width=0.9\linewidth]{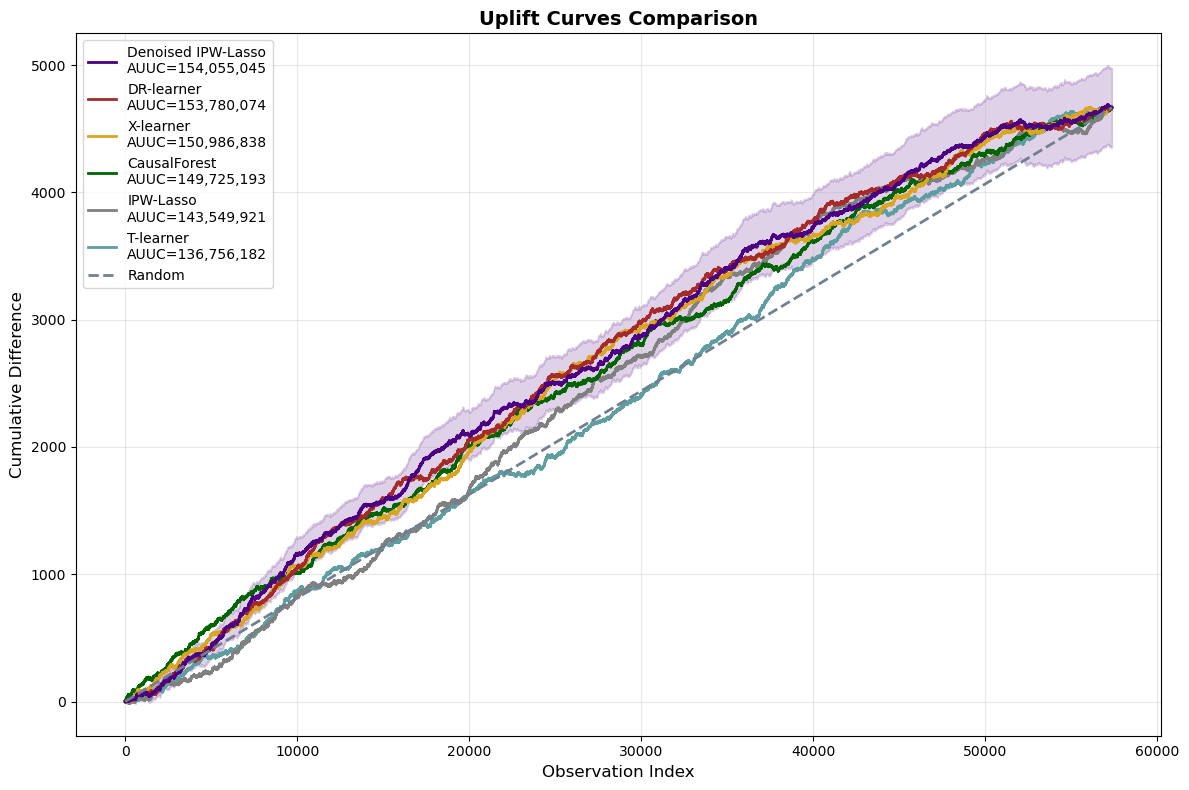}
    \caption{Uplift curve.}
    \label{fig:gotvuplift}
\end{figure}

Figure \ref{fig:gotvuplift} presents the uplift curves for all estimation methods. As shown, the DIPW-Lasso and DR-learner, both of which impose linearity on CATE, perform better than methods that do not impose any functional form.

Figure \ref{fig:gotvuplift} also presents the 95\% point-wise confidence band of the uplift curve for the DIPW-Lasso. This curve demonstrates a significant gain over random assignment, underscoring the utility of the DIPW-Lasso for decision-making. For example, suppose that the test set represents our target population, and consider a scenario in which only half of the units (28,680 units) can be treated due to limited resources. Based on the point estimate of the uplift curve, targeting units using the DIPW-Lasso is estimated to mobilize 2,742 units, compared to 2,333 under random assignment. This represents an approximate 17.5\% improvement in policy effect under the same budget.

In this example, denoising plays a significant role. The $R^2$, computed using the method described in Section 2.4, is 0.414. This suggests that a substantial portion of the IPW-transformed outcome consists of noise.
This is reflected in the difference in AUUC between DIPW-Lasso and IPW-Lasso (without denoising).

\subsubsection{Discover Interpretable Heterogeneity}
An additional benefit of employing the DIPW-Lasso is that it is effective for exploring the underlying sources of treatment effect heterogeneity by examining the estimated coefficients.

%\newpage
\begin{table}[htbp]
  \centering
  \begin{tabular}{l S[table-format=+1.8]}
    \toprule
    \textbf{Variable}   & \textbf{Coefficient} \\
    \midrule
    Intercept       & 0.0395 \\
    male            & -0.0058 \\
    age             & 0.0025 \\
    household size  & -0.0000 \\
    past\_voting\_1 & 0.0280 \\
    past\_voting\_2 & 0.0444 \\
    past\_voting\_3 & 0.0695 \\
    past\_voting\_4 & 0.0674 \\
    past\_voting\_5 & 0.0205 \\
    \bottomrule
  \end{tabular}
  \caption{Estimated coefficients.}
  \label{tab:GOTVcoef}
\end{table}

Table \ref{tab:GOTVcoef} lists the coefficient estimates obtained by the DIPW-Lasso.
All coefficients were estimated as nonzero, although that for household size is negligibly small.
The results are generally consistent with previous studies.
For instance, the estimate for age suggests that social pressure induces stronger treatment effects among individuals from older cohorts while having only a moderate impact on younger
individuals. This finding agrees with the prior literature, such as \textcite{Panagopoulos_et_al2014ElectoralStudies} and \textcite{fifield2019}, which attributed this pattern to life-cycle effects and age-related heterogeneity, respectively.

We also find clear heterogeneity in the treatment effect with respect to voting history.
Individuals who voted in three or four of the past five elections respond most positively to the treatment. In contrast, those who voted in all five previous elections exhibit only modest effects. The original study by \textcite{gerber2008social} did not reject the null hypothesis of equal treatment effects across different voting histories. However, other voter mobilization studies (e.g., \cite{niven2001limits,niven2004mobilization}) and subsequent meta-analytic findings from \textcite{arceneaux2009AJPS} have suggested that individuals' voting histories are an important source of treatment effect heterogeneity. Based on the literature, one possible explanation is that individuals who rarely or consistently voted are highly likely to maintain their habitual behavior in future elections, leaving limited room for mobilization. In contrast, those with intermediate voting histories can respond more to mobilization efforts. 

To assess whether the DIPW-Lasso appropriately captures systematic heterogeneity, we conducted an additional validation procedure. Specifically, we partitioned the test set into subgroups based on the values of the variables selected by the DIPW-Lasso. Then, we examined whether the subgroup ATE estimates  align with  the heterogeneity captured by the DIPW-Lasso. 

Figure \ref{fig: age_votingcount} presents the ATE estimates for subgroups defined by age and voting history. For the age variable, we divided the test set into four equally sized bins in increasing order of age. The subgroup effects exhibit a monotonic increase across the age bins, which is consistent with the positive coefficients for age, although the Lasso estimates its marginal effect. The ATE estimates  calculated for subgroups based on voting history are also broadly consistent with the coefficients estimated by the Lasso.

\begin{figure}[htbp]
    \centering
    \includegraphics[width=0.9\linewidth]{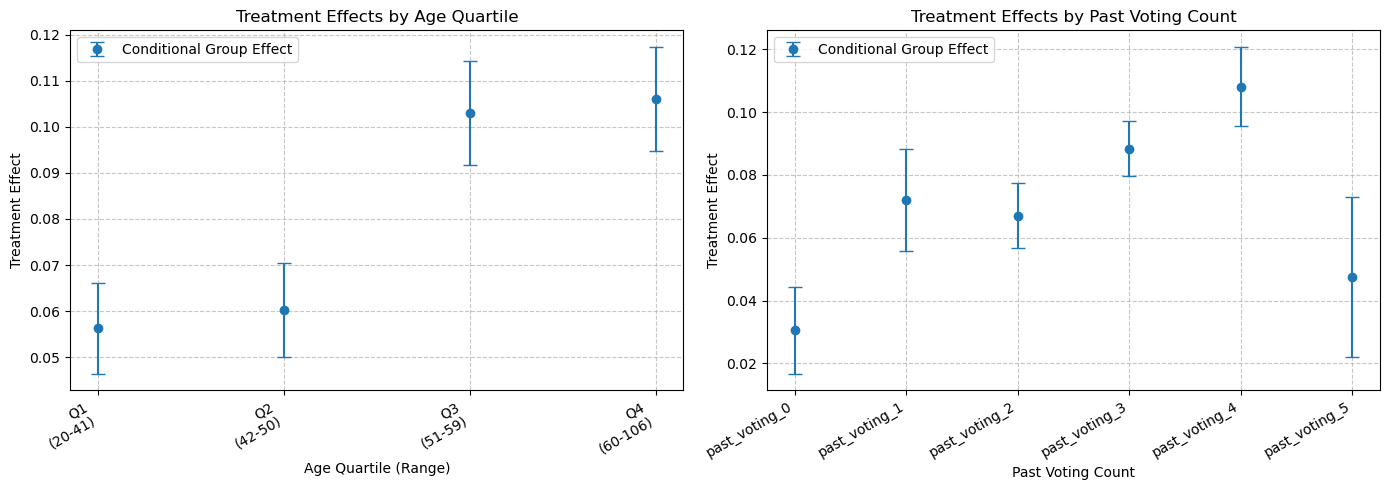}
    \caption{Group average treatment effects by age and past voting history.}
    \label{fig: age_votingcount}
\end{figure}

\subsection{Digital Advertising A/B Test}
In the second example, we use the Criteo Uplift Modeling dataset, a large-scale experimental dataset constructed from online advertising experiments \parencite{diemert2018large}. The dataset consists of anonymized units, and for each these, we observe whether they were exposed to an advertisement (treatment) and whether they visited the website (outcome).
Most of the units were randomly assigned to the treatment group, at a treatment probability of 0.85.

Our objective is to estimate the CATE of advertising on a binary outcome ``visit", which indicates whether a user visited the website. Because the dataset is huge (13,979,592 rows), we construct a manageable subset by drawing a stratified random sample of 100,000 rows, with stratification based on the joint distribution of the treatment indicator and the outcome. Subsequent analyses are performed on this sub-sample. 

In this application, we focus on comparing the ranking ability of different methods because the dataset anonymized individual features through random projection, thereby making the covariates no longer interpretable. We estimate the CATE using all 12 anonymized covariates provided by the dataset.

\begin{figure}[htbp]
    \centering
    \includegraphics[width=0.9\linewidth]{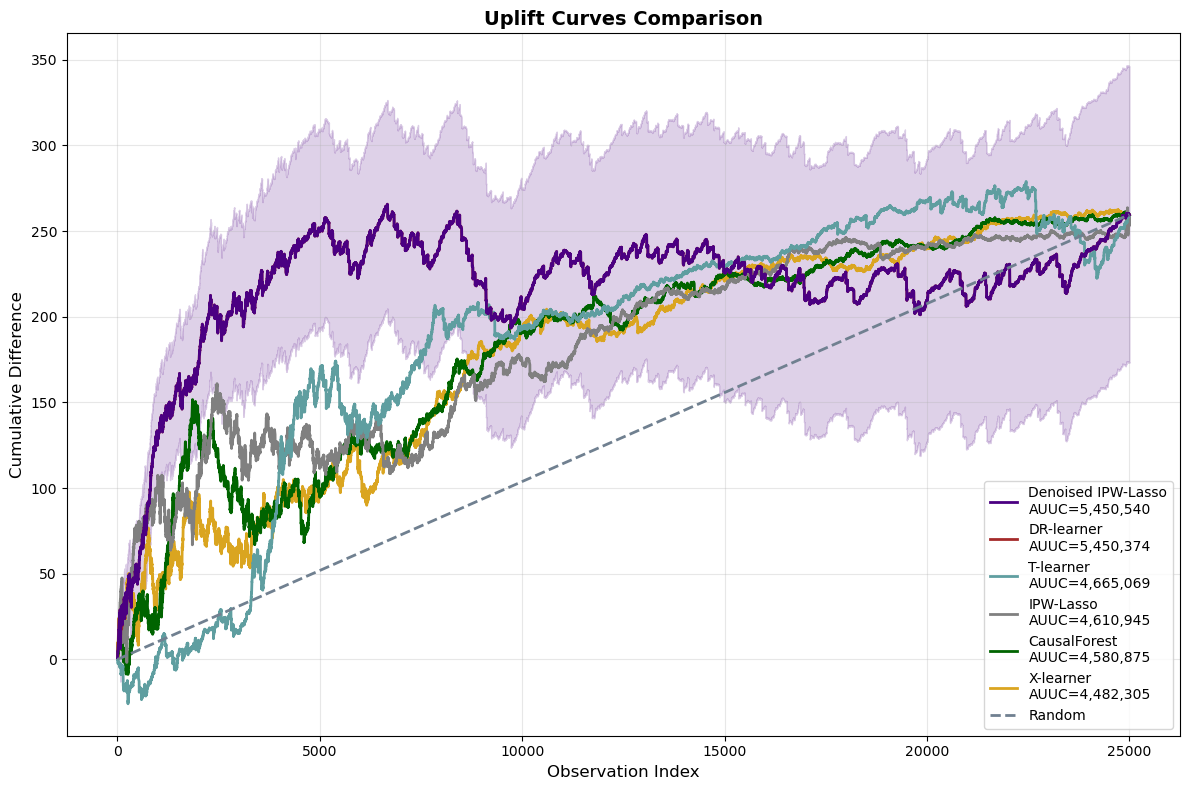}
    \caption{Uplift curve. }
    \label{fig:CriteoUpliftCurve}
\end{figure}

As shown in Figure~\ref{fig:CriteoUpliftCurve}, the DIPW-Lasso and DR-learner performed nearly equally well and achieved the highest AUUC.
This result also demonstrates the usefulness of assuming linearity for CATE.
The DIPW-Lasso now  selected only 2 out of the 12 covariates. The $R^2$ value from the denoising regression is 0.341.

Advertisers can leverage estimated CATE to optimize decision-making.
Treating all 25,000 units in the test set is estimated to yield an increase of 260 visits.
In contrast, targeting only the top 50\% ranked by the DIPW-Lasso results in an increase of 236 visits, achieving approximately 90.8\% of the total treatment effect at half the treatment cost.
Put differently, this targeting strategy yields an 81.5\% improvement over random assignment, which leads to only 130 additional visits when applied to 12,500 units.

\section{Conclusion}\label{sec:conclusion}
In the present paper, we propose  the DIPW-Lasso, an accurate and interpretable method for estimating heterogeneous treatment effects from RCT data, which has direct relevance to real-world decision-making. The method comprises two components: (i) a noise reduction technique that improves the efficiency of estimators based on IPW transformation and (ii) Lasso regularization to obtain a sparse linear CATE model. 
Our results provide the first theoretical guarantee that denoising reduces the prediction error in CATE estimation.

We assessed the effectiveness of our method through empirical analyses on two widely used datasets, the GOTV dataset and the Criteo Uplift Modeling dataset. Empirical studies demonstrated that the DIPW-Lasso substantially improves upon its non-denoised counterpart and achieves higher targeting accuracy than recent machine learning alternatives while maintaining a simple model structure, demonstrating its strong practical relevance.

Despite the above advantages, an important extension remains to be constructed . Although the DIPW-Lasso provides meaningful insights into treatment effect heterogeneity, the use of the Lasso renders conventional statistical inference infeasible. As this inferential limitation may restrict the translation of the data-driven findings into verifiable scientific conclusions, addressing inference in this context is a valuable direction for future research.

\section*{Acknowledgements}
We thank Ryosuke Hyodo, Hiro Kasahara, Tatsushi Oka, Takahide Yanagi, Kaiji Motegi, Akio Namba, and Mototsugu Shintani for their valuable comments.
We also thank participants at the Kansai Econometrics Meeting 2025, Japanese Joint Statistical Meeting 2024, Summer Workshop on Economic Theory, and 19th Japanese Statistical Society Spring Meeting for helpful comments and suggestions.

\section*{Funding}
Mingqian Guan is supported by JST SPRING Grant Number JPMJSP2148.
Naoya Sueishi is supported by JSPS KAKENHI Grant Number 24K04819.

\section*{Conflict of Interest}
Naoya Sueishi has a consulting agreement with CyberAgent, Inc., and Mingqian Guan previously interned at the same company. Shota Yasui and Komei Fujita are currently employed by CyberAgent, Inc. The authors declare no other conflicts of interest.

% References
{\normalem
\printbibliography
}

%%%%%%%%%%%%%%%%%%%%%%%%%%%%%%%
\appendix
\section{Appendix}

\subsection{Proof of Proposition \ref{proposition:optimalB}}
\begin{proof}
Expanding $E\left[(Y_i W_i - B(X_i) W_i)^2 \mid X_i = x\right]$ gives
\begin{align}\label{obj}
    &E\left[(Y_i W_i - B(X_i) W_i)^2 \mid X_i = x\right] \notag \\
    &\quad= E\left[Y_i^2 W_i^2 \mid X_i = x\right] - 2 B(x)\,E\left[Y_i W_i^2 \mid X_i = x\right] 
    + B^2(x)\,E\left[W_i^2 \mid X_i = x\right].
\end{align}
Additionally, conditional expectations are calculated as follows:
\begin{align*}
E[Y_i W_i^2 \mid X_i=x] &= \frac{\mu_1(x)}{p(x)} + \frac{\mu_0(x)}{1-p(x)}, \\
E[W_i^2 \mid X_i=x] &= \frac{1}{p(x)(1-p(x))}.
\end{align*}
Therefore, for each $x$, the value of $B(x)$ that minimizes \eqref{obj} is given by
\begin{equation*}
    B(x) = (1 - p(x)) \mu_1(x) + p(x) \mu_0(x),
\end{equation*}
completing the proof.
\end{proof}

%%%%%%%%%%%%%%%%%%%%%%%%%%%%%%
\subsection{Proof of Proposition \ref{proposition: basic_lassobound}}\label{proof:basciBound}
\begin{proof}
We first prove the concentration inequality \eqref{eq:noise_concentration}. By applying the Nemirovski moment inequality (see Lemma 14.24 in \cite{buhlmann2015high}) and using Assumptions \ref{assum:BoundednessofX} and \ref{assum:finiteL4moment}, we obtain
\begin{equation*}
E\left[\max_{1 \leq j \leq p}\left|\frac{e'X^{j}}{n}\right|^2\right]
\le 8 C^2 \sigma_e^2 \frac{\log 2p}{n}.
\end{equation*}
Furthermore, by the Markov inequality, we have
\begin{equation*}
P\left(\max_{1 \le j \leq p}\left|\frac{e'X^j}{n}\right|>c\right)
\leq \frac{E\left[\max_{1 \le j \leq p}\left|\frac{e' X^j}{n}\right|\right]}{c}
\leq \frac{\sqrt{E\left[\max_{1 \leq j \leq p} \left|\frac{e'X^j}{n}\right|^2\right]}}{c}
\leq \frac{\sqrt{8}C\sigma_e\sqrt{\log2p/n}}{c}
\end{equation*}
for any $c>0$.
Therefore, choosing $c = \sqrt{8} C \sigma_e\sqrt{\log 2p/n}/\eta$ for $\eta \in (0,1)$, we obtain
\begin{equation*}\label{eq:noise_concentrationProof}
    P\left( \left\| \frac{1}{n} X'e \right\|_{\infty} 
    \leq \frac{\sqrt{8} C \sigma_e \sqrt{\frac{\log 2p}{n}}}{\eta} \right) 
    \geq 1 - \eta.
\end{equation*}
Finally, let $\lambda = \frac{\sqrt{8} C \sigma_e \sqrt{\log 2p/n}}{\eta}$. 
Then, using \eqref{eq:BasicInequality}, we obtain
\begin{equation*}
    \frac{1}{n} \| X \tilde{\beta} - X \beta \|_2^2 \leq 4\lambda \| \beta \|_1
\end{equation*}
with probability at least \(1 -\eta\).
\end{proof}

To prove Theorem \ref{main_theorem}, we will need the following lemmas.

\begin{lemma}\label{lemma1} Suppose that Assumptions \ref{assum:overlap} and \ref{assum:BoundednessofX}--\ref{assum:NonSingularity} hold. Then, we have
    \begin{equation*}
        \left\| \frac{1}{n} X^{\prime} P_{\hat{Z}} Y^* \right\|_{\infty} = o_p\left( \sqrt{\frac{\log p}{n}} \right).
    \end{equation*}
\end{lemma}

\begin{proof}
We begin by decomposing the term as follows:
\begin{equation*}
\left\|\frac{1}{n} X^{\prime} P_{\hat{Z}} Y^*\right\|_{\infty} 
\leq 
\left\|\frac{1}{n} X^{\prime} \hat{Z}\right\|_{\infty} 
\left\|\left( \frac{1}{n} \hat{Z}^{\prime} \hat{Z} \right)^{-1}\right\|_2
\left\|\frac{1}{n} \hat{Z}^{\prime} Y^*\right\|_2
\equiv I_1 I_2 I_3.
\end{equation*}
The desired result follows by showing that $I_1 =O_p(\sqrt{\log p/n})$, $I_2=O_p(1)$, and $I_3 =o_p(1)$.

Next, decomposing $I_1$, we obtain
\begin{align*}
\left\|\frac{1}{n} X^{\prime} \hat{Z}\right\|_{\infty} 
&\leq \left\|\frac{1}{n} X^{\prime} Z\right\|_{\infty} + \left\|\frac{1}{n} X^{\prime} (\hat{Z} - Z)\right\|_{\infty} \notag \\
&\leq \left\|\frac{1}{n} X^{\prime} Z\right\|_{\infty} 
+ \frac{1}{K} \sum_{k=1}^{K}  \max_{1 \leq j \leq p}   \left| \frac{K}{n} \sum_{i \in I_k} X_{ij} W_i (\widehat{B}_{I^c_k}(X_i)-B(X_i)) \right|,
\end{align*}
where the last inequality follows from the fact that the all elements in the first column of $\hat{Z} - Z$ are zero.
Moreover, we have
\begin{equation*}
\left\| \frac{1}{n} X^{\prime} Z \right\|_{\infty} 
\le
\max_{1 \leq j \leq p} \left| \frac{1}{n} \sum_{i=1}^n X_{ij} W_i \right| +
\max_{1 \leq j \leq p} \left| \frac{1}{n} \sum_{i=1}^n X_{ij} W_i B(X_i) \right|.
\end{equation*}
Since $E[X_{ij}W_i B(X_i)]=0$ for $j=1, \dots, p$ and any function $B(\cdot)$, applying Nemirovski's inequality and Assumption \ref{assum:BoundednessofX} yields
\begin{align*}
E\left[\max_{\substack{1 \leq j \leq p}} \left| \frac{1}{n} \sum_{i=1}^n X_{ij} W_{i}B(X_i) \right|^2\right] & \leq \frac{8\log(2p)}{n} C^2  E\left[ W_i^2 B^2(X_i) \right].
\end{align*}
Note that the above inequality also holds for $B(X_i)=1$.
Thus, the Markov inequality implies $\| n^{-1} X'Z\|_\infty=O_p(\sqrt{\log p/n})$.
Furthermore, for $i \in I_k$, we have
\[
 E \left[ X_{ij} W_i \left( \widehat{B}_{I^c_k}(X_i) - B(X_i) \right) \mid I_k^c\right] =0.
\]
Since the random variables $\{X_{ij} W_i ( \widehat{B}_{I^c_k}(X_i) - B(X_i))\}_{i \in I_k}$ are independent conditional on $I_k^c$, we can apply Nemirovski's inequality again to obtain
\begin{align*}
& E \left[ \max_{1 \leq j \leq p} \left| \frac{K}{n} \sum_{i \in I_k} X_{ij} W_i \left( \widehat{B}_{I^c_k}(X_i) - B(X_i) \right) \right|^2 \mid I_k^c \right] \\
& \quad \leq \frac{8 K \log(2p)}{n} C^2 E\left[ W_i^2 \left( \widehat{B}_{I^c_k}(X_i) - B(X_i) \right)^2 \mid I_k^c \right].
\end{align*}
Moreover, from  Assumption \ref{assum:overlap}, we have $|W_i|< 1/\xi < \infty$.
Thus, we obtain
\begin{align*}
&E \left[ \frac{1}{K} \sum_{k=1}^{K}  \max_{1 \leq j \leq p}   \left| \frac{K}{n} \sum_{i \in I_k} X_{ij} W_i (\widehat{B}_{I^c_k}(X_i)-B(X_i)) \right| \right] \\
& \quad \le \frac{8 K \log(2p)}{n} C^2 \xi^{-2} E\left[ \left( \widehat{B}_{I^c_k}(X_i) - B(X_i) \right)^2\right],
\end{align*}
which implies
\begin{align*}
   \frac{1}{K} \sum_{k=1}^K \max_{1 \leq j \leq p} \left| \frac{K}{n} \sum_{i \in I_k} X_{ij} W_i (\widehat{B}_{I^c_k}(X_i)-B(X_i)) \right| &= o_p\left( \sqrt{ \frac{ \log p }{n} } \right)
\end{align*}
by Assumption \ref{assum:l2convergence}.
Therefore, we obtain $I_1 =O_p(\sqrt{\log p/n})$.

Next, $n^{-1} \hat{Z}'\hat{Z}$ can be written as
\begin{align*}
\frac{1}{n} \hat{Z}'\hat{Z} 
= \begin{pmatrix}
n^{-1} \sum_{i=1}^n W_i^2 & n^{-1} \sum_{k=1}^K \sum_{i \in I_k} W_i^2 \widehat{B}_{I^c_k}(X_i) \\ 
n^{-1} \sum_{k=1}^K \sum_{i \in I_k} W_i^2 \widehat{B}_{I^c_k}(X_i) & n^{-1} \sum_{k=1}^K \sum_{i \in I_k} W_i^2 \widehat{B}^2_{I^c_k}(X_i)
\end{pmatrix}.
\end{align*}
It follows from the Markov inequality and Assumptions \ref{assum:overlap} and \ref{assum:l2convergence} that
\begin{align*}
P\left( \left| \frac{1}{n} \sum_{k=1}^K \sum_{i \in I_k} W_i^2 \left( \widehat{B}_{I^c_k}(X_i) - B(X_i) \right) \right| > \epsilon \right)  
& \leq  \frac{ E \left[ \left|\widehat{B}_{I_k^c}(X_i)-B(X_i) \right| \right]}{\xi^2 \epsilon} \to 0
\end{align*}
for any $\epsilon >0$. Hence, we have
\begin{align*}
\frac{1}{n} \sum_{k=1}^K \sum_{i \in I_k} W_i^2 \widehat{B}_{I^c_k}(X_i) 
= \frac{1}{n} \sum_{i=1}^n W_i^2 B(X_i) + o_p(1).
\end{align*}
Similarly, we obtain
\[
P \left( \left|\frac{1}{n} \sum_{k=1}^K \sum_{i \in I_k} W_i^2 \left( \widehat{B}_{I^c_k}(X_i) - B(X_i) \right)^2 \right|> \epsilon \right)
\leq 
\frac{E\left[\left(\widehat{B}_{I_k^c}(X_i)-B(X_i) \right)^2 \right]}{\xi^2 \epsilon} \to 0.
\]
Moreover, by the Cauchy-Schwarz inequality, we have 
\begin{align*}
&\left|\frac{1}{n} \sum_{k=1}^K \sum_{i \in I_k} W_i^2 B(X_i) \left( \widehat{B}_{I^c_k}(X_i) - B(X_i) \right) \right| \\
& \quad \le \xi^{-2} \sqrt{\frac{1}{n} \sum_{i=1}^n B^2
(X_i)} \sqrt{\frac{1}{n} \sum_{k=1}^K \sum_{i \in I_k} \left(\widehat{B}_{I_k^c}(X_i)-B(X_i) \right)^2}= o_p(1).
\end{align*}
Thus, we have
\begin{align*}
\frac{1}{n} \sum_{k=1}^K \sum_{i \in I_k} W_i^2 \widehat{B}^2_{I^c_k}(X_i) 
&= \frac{1}{n} \sum_{i=1}^n W_i^2 B^2(X_i)
+ \frac{1}{n} \sum_{k=1}^n \sum_{i \in I_k} W_i^2 \left( \widehat{B}_{I^c_k}(X_i) - B(X_i) \right)^2 \\
& \quad + \frac{2}{n} \sum_{k=1}^n \sum_{i \in I_k} W_i^2 B(X_i) \left( \widehat{B}_{I^c_k}(X_i) - B(X_i) \right) \\
& = \frac{1}{n} \sum_{i \in I_k} W_i^2 B^2(X_i) +o_p(1).
\end{align*}
Therefore, by the law of large numbers, we have
$n^{-1}\hat{Z}_k^{\prime} \hat{Z}_k = E[Z_i Z_i^{\prime}] + o_p(1)$, which implies that $I_2=O_p(1)$.

Finally, $n^{-1} \hat{Z}' Y^*$ can be written as 
\begin{equation*}
n^{-1} \hat{Z}' Y^* = 
\begin{pmatrix}
n^{-1} \sum_{i=1}^n W_i Y^*_i \\[0.5em]
n^{-1} \sum_{k=1}^K \sum_{i \in I_k} W_i \widehat{B}_{I^c_k}(X_i) Y^*_i 
\end{pmatrix}.
\end{equation*}
Because $E[W_i Y_i^*]=E[W_i B(X_i) Y_i^*]=0$ from the definition of the linear projection, using similar arguments as those above, we obtain $I_3=o_p(1)$. 
\end{proof}

\begin{lemma}\label{lemma2} Suppose that Assumptions \ref{assum:overlap} and \ref{assum:BoundednessofX}--\ref{assum:NonSingularity} hold. Then, we have
    \begin{equation*}
        \left\| \frac{1}{n} X^{\prime} \left(Z - P_{\hat{Z}} Z\right) \gamma \right\|_{\infty} = o_p\left( \sqrt{\frac{\log p}{n}} \right).
    \end{equation*}
\end{lemma}

\begin{proof}
We bound $\left\| n^{-1} X^{\prime} \left(Z - P_{\hat{Z}} Z\right) \gamma \right\|_{\infty}$ as 
\begin{align*}
&\left\| \frac{1}{n} X^\prime (Z - P_{\hat{Z}} Z) \gamma \right\|_{\infty} \\
&\leq \left\| \frac{1}{n} X^\prime Z \gamma 
- \frac{1}{n} X' P_{\hat{Z}}\hat{Z} \gamma  \right\|_{\infty} 
+ \left\| \frac{1}{n} X^\prime P_{\hat{Z}} (Z - \hat{Z}) \gamma \right\|_{\infty} \\
&= \left\| \frac{1}{n} X^\prime (Z - \hat{Z}) \gamma \right\|_{\infty} 
+ \left\| \frac{1}{n} X^\prime \hat{Z} \right\|_{\infty} 
\left\| \left( \frac{1}{n} \hat{Z}^\prime \hat{Z} \right)^{-1} \right\|_2 
\left\| \frac{1}{n} \hat{Z}^\prime (Z - \hat{Z}) \gamma \right\|_2.
\end{align*}
Here, we have
\begin{equation*}
\left\|\frac{1}{n} \hat{Z}^\prime (Z-\hat{Z}) \gamma \right\|_2 = \left\|\begin{array}{l}
n^{-1} \sum_{k=1}^K \sum_{i\in I_k} W_i^2\left(B\left(X_i\right)-\widehat{B}_{I^c_k}\left(X_i\right)\right) \alpha_2 \\
n^{-1} \sum_{k=1}^K \sum_{i\in I_k} W_i^2 \widehat{B}_{I^c_k}\left(X_i\right)\left(B\left(X_i\right)-\widehat{B}_{I^c_k}\left(X_i\right)\right) \alpha_2
\end{array}\right\|_2.
\end{equation*}
Moreover, we have
\begin{align*}
&\left| \frac{1}{n} \sum_{k=1}^K \sum_{i \in I_k} W_i^2 \widehat{B}_{I^c_k}(X_i) \left( B(X_i) - \widehat{B}_{I^c_k}(X_i) \right) \right| \\
&\quad  \le \left| \frac{1}{n} \sum_{k=1}^K \sum_{i \in I_k} W_i^2 \left( \widehat{B}_{I^c_k}(X_i)-B(X_i) \right)^2  \right|
+ \left|
\frac{1}{n} \sum_{k=1}^K \sum_{i \in I_k}  W_i^2 B(X_i) \left(  \widehat{B}_{I^c_k}(X_i)-B(X_i)\right) \ \right|
\end{align*}
Thus, by the proof of Lemma \ref{lemma1}, we have $\| n^{-1} \hat{Z}^\prime (Z - \hat{Z}) \gamma \|_2 = o_p(1)$, which also implies the desired result.
\end{proof}

\subsection{Proof of Theorem \ref{main_theorem}}

\begin{proof}
By a similar argument as in the proof of Proposition \ref{proposition: basic_lassobound}, we can obtain
\begin{equation*}
    P\left( \left\| \frac{1}{n} X'u \right\|_{\infty} 
    \leq \frac{\sqrt{8} C \sigma_u \sqrt{\frac{\log 2p}{n}}}{\eta} \right) 
    \geq 1 - \eta
\end{equation*}
for $\eta \in (0,1)$.
Therefore, by applying the triangular inequality, we obtain
\begin{align*}
&P\left( \left\| \frac{1}{n} X' \left( u + \left(Z - P_{\hat{Z}} Z\right) \gamma - P_{\hat{Z}} Y^* \right) \right\|_{\infty} 
\geq \left( \frac{\sqrt{8} C \sigma_u}{\eta} + \epsilon \right) \sqrt{\frac{\log 2p}{n}} \right) \notag \\
&\quad \leq P\left( \left\| \frac{1}{n} X'u \right\|_{\infty} 
\geq \frac{\sqrt{8} C \sigma_u}{\eta} \sqrt{\frac{\log 2p}{n}} \right) \notag \\
&\quad \quad + P\left( \left\|\frac{1}{n} X^{\prime} \left(Z - P_{\hat{Z}} Z \right) \gamma + \frac{1}{n} X^{\prime} P_{\hat{Z}} Y^*  \right\|_{\infty} 
\geq \epsilon \sqrt{\frac{\log 2p}{n}} \right) \notag \\
&\quad \leq \eta + P\left( \left\|\frac{1}{n} X^{\prime} \left(Z - P_{\hat{Z}} Z \right) \gamma + \frac{1}{n} X^{\prime} P_{\hat{Z}} Y^*  \right\|_{\infty} 
\geq \epsilon \sqrt{\frac{\log 2p}{n}} \right)
\end{align*}
for any $\epsilon >0$.
The desired result follows from Lemmas \ref{lemma1} and \ref{lemma2} and \eqref{eq:BasicInequalityofDIL}.
\end{proof}

\end{document}